\documentclass[11pt]{article}

\usepackage{color}
\usepackage{amsmath,amsfonts,amsthm}
\usepackage{hyperref}
\usepackage{mathtools}
\usepackage{lmodern}
\usepackage{multirow}
\usepackage{array}
\usepackage{microtype}

\usepackage{fullpage}
\usepackage{amsmath,amsthm,amsfonts,dsfont}
\usepackage{amssymb,latexsym,graphicx}
\usepackage{mathtools}
\usepackage{hyperref}
\usepackage{xcolor}
\hypersetup{
	colorlinks,
	linkcolor={red!75!black},
	citecolor={blue!75!black},
	urlcolor={blue!75!black}
}

\newtheorem{theorem}{Theorem}[section]
\newtheorem{proposition}[theorem]{Proposition}
\newtheorem{lemma}[theorem]{Lemma}
\newtheorem{claim}[theorem]{Claim}

\newtheorem{corollary}[theorem]{Corollary}
\newtheorem{definition}[theorem]{Definition}

\newcommand{\R}{\ensuremath{\mathbb{R}}}
\newcommand{\Z}{\ensuremath{\mathbb{Z}}}
\newcommand{\Q}{\mathbb{Q}}

\newcommand{\lat}{\mathcal{L}}

\newcommand{\eps}{\varepsilon} 
\renewcommand{\epsilon}{\varepsilon}
\newcommand{\poly}{\mathrm{poly}}
\DeclareMathOperator{\dist}{dist}

\renewcommand{\vec}[1]{\ensuremath{\boldsymbol{#1}}}
\newcommand{\problem}[1]{\ensuremath{\mathrm{#1}} }

\DeclarePairedDelimiter\inner{\langle}{\rangle}
\newcommand{\CVP}{\problem{CVP}}
\newcommand{\SVP}{\problem{SVP}}
\newcommand{\CVPP}{\problem{CVPP}}

\newcommand{\norm}[1]{\|#1\|}
\newcommand{\set}[1]{\{#1\}}
\newcommand{\size}{\textrm{size}}
\newcommand{\ind}{\textrm{ind}}
\newcommand{\val}{\textrm{val}}
\newcommand{\half}{\frac{1}{2}}

\newcommand{\cc}[1]{\ensuremath{\mathsf{#1}}}
\newcommand{\SP}{\cc{\#P}}

\begin{document}

\title{On the Quantitative Hardness of CVP}
\author{
	Huck Bennett\thanks{Courant Institute of Mathematical Sciences, New York
 University.}\\ 
\texttt{huckbennett@gmail.com}
	\and
	Alexander Golovnev\thanks{Yahoo Research.}~\thanks{Part of this work was done while the author was at Courant Institute of Mathematical Sciences, and was supported by the National Science Foundation under Grant No. CCF-1320188. Any opinions, findings, and conclusions or recommendations expressed in this material are those of the
authors and do not necessarily reflect the views of the National Science Foundation.}\\
\texttt{alexgolovnev@gmail.com}
 \and
	Noah Stephens-Davidowitz\footnotemark[1]~\thanks{Supported by the National Science Foundation (NSF) under Grant No.~CCF-1320188, and the Defense Advanced Research Projects Agency (DARPA) and Army Research
		Office (ARO) under Contract No.~W911NF-15-C-0236.}\\
	\texttt{noahsd@gmail.com}
}
\date{}
\maketitle
	
\begin{abstract}
For odd integers $p \geq 1$ (and $p = \infty$), we show that the Closest Vector Problem in the $\ell_p$ norm ($\CVP_p$) over rank $n$ lattices cannot be solved in $2^{(1-\eps) n}$ time for any constant $\eps > 0$ unless the Strong Exponential Time Hypothesis (SETH) fails. We then extend this result to ``almost all'' values of $p \geq 1$, not including the even integers. 
This comes tantalizingly close to settling the quantitative time complexity of the important special case of $\CVP_2$ (i.e., $\CVP$ in the Euclidean norm), for which a $2^{n +o(n)}$-time algorithm is known. In particular, our result applies for any $p = p(n) \neq 2$ that approaches $2$ as $n \to \infty$.

We also show a similar SETH-hardness result for $\SVP_\infty$; hardness of approximating $\CVP_p$ to within some constant factor under the so-called Gap-ETH assumption; and other quantitative hardness results for $\CVP_p$ and $\CVPP_p$ for any $1 \leq p < \infty$ under different assumptions.
\end{abstract}

\setcounter{page}{0}
\thispagestyle{empty}
\newpage
\setcounter{page}{1}
	
\section{Introduction}
\label{sec:intro}	
A lattice $\lat$ is the set of all integer combinations of
linearly independent basis vectors $\vec{b}_1,\dots,\vec{b}_n \in \R^d$,
\[
\lat = \lat(\vec{b}_1,\ldots, \vec{b}_n) := \Big\{ \sum_{i=1}^n z_i \vec{b}_i \ : \ z_i \in \Z \Big\}
\; .
\] We call $n$ the \emph{rank} of the lattice $\lat$ and $d$ the \emph{dimension} or the \emph{ambient dimension}.

The two most important computational problems on lattices are the Shortest Vector Problem ($\SVP$) and the Closest Vector Problem ($\CVP$). Given a basis for a lattice $\lat \subset \R^d$,
$\SVP$ asks us to compute the minimal length of a non-zero vector in $\lat$, and $\CVP$ asks us to compute the distance from some target point $\vec{t} \in \R^d$ to the lattice. Typically, we define length and distance in terms of the $\ell_p$ norm for some $1 \leq p \leq \infty$, given by
\[
\|\vec{x}\|_p := (|x_1|^p + |x_2|^p + \cdots + |x_d|^p)^{1/p}
\]
for finite $p$ and
\[
\|\vec{x}\|_{\infty} := \max_{1 \leq i \leq d} |x_i|
\; .
\]
In particular, the $\ell_2$ norm is the familiar Euclidean norm, and it is by far the best studied in this context.
We write $\SVP_p$ and $\CVP_p$ for the respective problems in the $\ell_p$ norm. $\CVP$ is known to be at least as hard as $\SVP$ (in any norm, under an efficient reduction that preserves the rank and approximation factor)~\cite{GMSS99} and appears to be strictly harder.

Starting with the breakthrough work of Lenstra, Lenstra, and Lov{\'a}sz in 1982~\cite{LLL82}, algorithms for solving these problems in both their exact and approximate forms have found innumerable applications, including
factoring polynomials over the rationals~\cite{LLL82}, integer programming~\cite{Lenstra83,Kannan87,DPV11}, cryptanalysis~\cite{Shamir84,Odl90,JS98,NS01}, etc. More recently, many cryptographic primitives have been constructed whose security is based on the {\em worst-case} hardness of these or closely related lattice problems \cite{Ajtai96,Reg09,GPV08,Peikert08,chris_survey}. Given the obvious importance of these problems, their complexity is quite well studied. Below, we survey some of these results. We focus on algorithms for the exact and near-exact problems since these are most relevant to our work and because the best known algorithms for the approximate variants of these problems typically use algorithms for the exact problems as subroutines~\cite{Schnorr87,GN08,MW16}. (Many of the results described below are also summarized in Table \ref{tab:complexity_summary}.) 

\subsection{Algorithms for SVP and CVP}

\paragraph{The AKS algorithm and its descendants. } 
The current fastest known algorithms for solving $\SVP_p$ all use the celebrated randomized sieving technique due to Ajtai, Kumar, and Sivakumar~\cite{AKS01}. The original algorithm from~\cite{AKS01} was the first $2^{O(n)}$-time algorithm for $\SVP$, and it worked for both $p = 2$ and $p = \infty$.

In the $p = 2$ case, a sequence of works improved upon the constant in the exponent~\cite{NguyenVidick08,PS09,MV10,LWXZ11}, and the current fastest running time of an algorithm that provably solves $\SVP_2$ exactly is $2^{n + o(n)}$~\cite{ADRS15}.\footnote{The algorithm in~\cite{ADRS15} is quite a bit different than the other algorithms in this class, but it can still be thought of as a sieving algorithm.} While progress has slowed, this seems unlikely to be the end of the story. Indeed, there are heuristic sieving algorithms that run in time $(3/2)^{n/2 + o(n)}$~\cite{NguyenVidick08,WLTB11,Laarhoven2015,BDGL16}, and there is some reason to believe that the provably correct~\cite{ADRS15} algorithm can be improved. In particular, there is a provably correct $2^{n/2 + o(n)}$-time algorithm that approximates $\SVP_2$ up to a small constant approximation factor~\cite{ADRS15}.

A different line of work extended the randomized sieving approach of~\cite{AKS01} to obtain $2^{O(n)}$-time algorithms for $\SVP$ in additional norms. In particular, Bl{\"o}mer and Naewe extended it to all $\ell_p$ norms~\cite{BN09}. Subsequent work extended this further, first to arbitrary symmetric norms~\cite{AJ08} and then to the ``near-symmetric norms'' that arise in integer programming~\cite{Dadush2012}.

Finally, a third line of work extended the~\cite{AKS01} approach to approximate $\CVP$. Ajtai, Kumar, and Sivakumar themselves showed a $2^{O(n)}$-time algorithm for approximating $\CVP_2$ to within any constant approximation factor strictly greater than one~\cite{AKS02}. Bl{\"o}mer and Naewe obtained the same result for all $\ell_p$ norms~\cite{BN09}, and Dadush extended it further to arbitrary symmetric norms and again to ``near-symmetric norms''~\cite{Dadush2012}. We stress, however, that none of these results apply to exact $\CVP$, and indeed, there are some barriers to extending these algorithms to exact $\CVP$. (See, e.g.,~\cite{ADS15}.)

\paragraph{Exact algorithms for $\CVP$. } \emph{Exact} $\CVP$ appears to be a much more subtle problem than exact $\SVP$.%
\footnote{In particular, there can be arbitrarily many lattice points that are approximate closest vectors, which makes sieving techniques seemingly useless for solving exact $\CVP$. (See, e.g.,~\cite{ADS15} for a discussion of this issue.) We note, however, that hardness results (including ours) tend to produce $\CVP$ instances with a bounded number of approximate closest vectors (e.g., $2^{O(n)})$.}
Indeed, progress on exact $\CVP$ has been much slower than the progress on exact $\SVP$.
Over a decade after~\cite{AKS01}, Micciancio and Voulgaris presented the first $2^{O(n)}$-time algorithm for exact $\CVP_2$~\cite{MV13}, using elegant new techniques built upon the approach of Sommer, Feder, and
Shalvi~\cite{SFS09}. Specifically, they achieved a running time of $4^{n + o(n)}$, and subsequent work even showed a running time of $2^{n + o(n)}$ for $\CVP_2$ with Preprocessing (in which the algorithm is allowed access to arbitrary advice that depends on the lattice but not the target vector; see Section~\ref{sec:problem_defs})~\cite{BonifasD14}. Later,~\cite{ADS15} showed a $2^{n + o(n)}$-time algorithm for $\CVP_2$, so that the current best known asymptotic running time is actually the same for $\SVP_2$ and $\CVP_2$. 

However, for $p \neq 2$, progress for exact $\CVP_p$ has been minimal. Indeed, the fastest known algorithms for exact $\CVP_p$ with $p \neq 2$ are still the $n^{O(n)}$-time enumeration algorithms first developed by Kannan in 1987~\cite{Kannan87,DPV11,MicciancioWalter14}.  Both algorithms for exact  $\CVP_2$ mentioned in the previous paragraph use many special properties of the $\ell_2$ norm, and it seems that substantial new ideas would be required to extend them to arbitrary $\ell_p$ norms. 

\subsection{Hardness of SVP and CVP}

Van Emde Boas showed the NP-hardness of $\CVP_p$ for any $p$ and $\SVP_\infty$ in 1981~\cite{Boas81}. Extending this to $\SVP_p$ for finite $p$ was a major open problem until it was proven (via a randomized reduction) for all $1 \leq p \leq \infty$ by Ajtai in 1998~\cite{Ajtai-SVP-hard}. There has since been much follow-up work, showing the hardness of these problems for progressively larger approximation factors, culminating in NP-hardness of approximating $\CVP_p$ up to a factor of $n^{c/\log \log n}$ for some constant $c > 0$~\cite{ABSS93,DKRS03} and hardness of $\SVP_p$ with the same approximation factor under plausible complexity-theoretic assumptions~\cite{CN98,Mic01svp,Khot05svp,HRsvp}. These results are nearly the best possible under plausible assumptions, since approximating either problem up to a factor of $\sqrt{n}$ is known to be in $\mathrm{NP} \cap \mathrm{coNP}$~\cite{GG00,AharonovR05, Peikert08}.

However, such results only rule out the possibility of polynomial-time algorithms (under reasonable complexity-theoretic assumptions). They say very little about the \emph{quantitative} hardness of these problems for a fixed lattice rank $n$.%
\footnote{
	One can derive certain quantitative hardness results from known hardness proofs, but in most cases the resulting lower bounds are quite weak. 
	The only true quantitative hardness results known prior to this work were a folklore ETH-hardness result for $\CVP$ and an unpublished result due to Samuel Yeom, showing that $\CVP$ cannot be solved in time $2^{10^{-4} n}$ under plausible complexity-theoretic assumptions~\cite{priv:Vinod}. (In Section~\ref{sec:ETH-MAX-CUT}, we present a similar proof of a stronger statement.)
	}

This state of affairs is quite frustrating for two reasons. First, in the specific case of $\CVP_2$, algorithmic progress has reached an apparent barrier. In particular, both known techniques for solving exact $\CVP_2$ in singly exponential time are fundamentally unable to produce algorithms whose running time is asymptotically better than the current best of $2^{n + o(n)}$~\cite{MV13,ADS15}.%
\footnote{
	Both techniques require short vectors in each of the $2^n$ cosets of $\lat$ mod $2\lat$ (though for apparently different reasons).
	} 
Second, some lattice-based cryptographic constructions are close to deployment~\cite{new_hope,frodo,NIST_quantum}. In order to be practically secure, these constructions require the quantitative hardness of certain lattice problems, and so their designers rely on quantitative hardness assumptions~\cite{APS15}. If, for example, there existed a $2^{n/20}$-time algorithm for $\SVP_p$ or $\CVP_p$, then these cryptographic schemes would be insecure in practice.

We therefore move in a different direction. Rather than trying to extend non-quantitative hardness results to larger approximation factors, we show quantitative hardness results for exact (or nearly exact) problems. To do this, we use the tools of \emph{fine-grained complexity}.

\subsection{Fine-grained complexity}
\label{sec:fine-grained}
Impagliazzo and Paturi~\cite{IP1999} introduced the \emph{Exponential Time Hypothesis} (ETH) and the \emph{Strong Exponential Time Hypothesis} (SETH) to help understand the precise hardness of $k$-SAT. Informally, ETH asserts that $3$-SAT takes $2^{\Omega(n)}$-time to solve in the worst case, and SETH asserts that $k$-SAT takes essentially $2^n$-time to solve for unbounded $k$. I.e., SETH asserts that brute-force search is essentially optimal for solving $k$-SAT for large $k$.

Recently, the study of fine-grained complexity has leveraged ETH, SETH, and several other assumptions to prove quantitative hardness results about a wide range of problems. These include both problems in $\mathsf{P}$ (see, e.g.,~\cite{conf/soda/ChechikLRSTW14, conf/stoc/BackursI15, conf/focs/AbboudBW15} and the survey by Vassilevska Williams~\cite{conf/iwpec/Williams15}), and $\mathsf{NP}$-hard problems (see, e.g.,~\cite{conf/soda/PatrascuW10, conf/coco/CyganDLMNOPSW12, books/sp/CyganFKLMPPS15}). Although these results are all conditional, they help to explain \emph{why} making further algorithmic progress on these problems is difficult---and suggest that it might be impossible. Namely, any non-trivial algorithmic improvement would disprove a very well-studied hypothesis.

One proves quantitative hardness results using \emph{fine-grained} reductions (see~\cite{conf/iwpec/Williams15} for a formal definition). For example, there is an efficient mapping from $k$-SAT formulas on $n$ variables to Hitting Set instances with universes of $n$ elements~\cite{conf/coco/CyganDLMNOPSW12}. This reduction is fine-grained in the sense that for any constant $\eps > 0$, a $2^{(1 - \eps)n}$-time algorithm for Hitting Set implies a $2^{(1 - \eps)n}$-time algorithm for $k$-SAT, breaking SETH.

Despite extensive effort, no faster-than-$2^n$-time algorithm for $k$-SAT with unbounded $k$ has been found. Nevertheless, there is no consensus on whether SETH is true or not, and recently, Williams~\cite{conf/coco/Williams16} refuted a very strong variant of SETH. This makes it desirable to base quantitative hardness results on weaker assumptions when possible, and indeed our main result holds even assuming a weaker variant of SETH based on the hardness of Weighted Max-$k$-SAT (except for the case of $p = \infty$).

\subsection{Our contribution}

We now enumerate our results. See also Table~\ref{tab:complexity_summary}.

\begin{table}
	\begin{center}
		\begin{tabular}{| c | c | c |  c | c | c | c |}
			\hline
			Problem & Upper Bound & \multicolumn{4}{|c|}{Lower Bounds} & Notes \\
			\hline
			&  & SETH & Max-$2$-SAT & ETH & Gap-ETH &  \\
			\hline
			$\CVP_p$  & $n^{O(n)}$ ($2^{O(n)}$) & {\color{blue!70!black} $2^n$} & {\color{blue!70!black} $2^{\omega n/3}$} & $2^{\Omega(n)}$ & {\color{blue!70!black} $2^{\Omega(n)}$*} & ``almost all'' $p \notin 2\Z$  \\
			$\CVP_2$ & $2^{n}$ & --- & {\color{blue!70!black}$2^{\omega n/3}$} & $2^{\Omega(n)}$ & {\color{blue!70!black} $2^{\Omega(n)}$*} & \\
			$\CVP_\infty$/$\SVP_\infty$ & $2^{O(n)}$ & {\color{blue!70!black}$2^n$*} & --- & $2^{\Omega(n)}$ &{\color{blue!70!black} $2^{\Omega(n)}$*} &  \\
			$\CVPP_p$ & $n^{O(n)}$ ($2^{O(n)}$) & --- & {\color{blue!70!black}$2^{\Omega(\sqrt{n}) }$} & {\color{blue!70!black}$2^{\Omega(\sqrt{n})}$} & --- & \\
			\hline
		\end{tabular}
		\caption{\label{tab:complexity_summary} Summary of known quantitative upper and lower bounds, with new results in {\color{blue!70!black}blue}. Upper bounds in parentheses hold for any constant approximation factor strictly greater than one, and lower bounds with a $*$ apply for some constant approximation factor strictly greater than one. $\omega$ is the matrix multiplication exponent, satisfying $2 \leq \omega < 2.373$. We have suppressed smaller factors. }
	\end{center}
\end{table}

\paragraph{SETH-hardness of $\CVP_p$. } Our main result is the SETH-hardness of $\CVP_p$ for any odd integer $p \geq 1$ and $p = \infty$ (and $\SVP_\infty$). Formally, we prove the following. (See Sections~\ref{sec:geom-char} and~\ref{sec:SETH} for finite $p$ and Section~\ref{sec:infinity!} for $p = \infty$.)

\begin{theorem}
	\label{thm:main}
	For any constant integer $k \geq 2$ and any odd integer $p \geq 1$ or $p = \infty$, there is an efficient reduction from $k$-SAT with $n$ variables and $m$ clauses to $\CVP_p$ (or $\SVP_\infty$) on a lattice of rank $n$ (with ambient dimension $n + O(m)$).

	In particular, there is no $2^{(1-\eps) n}$-time algorithm for $\CVP_p$ for any odd integer $p \geq 1$ or $p =\infty$ (or $\SVP_\infty$) and any constant $\eps > 0$ unless SETH is false.
\end{theorem}

Unfortunately, we are unable to extend this result to even integers $p$, and in particular, to the important special case of $p = 2$. In fact, this is inherent, as we show that our approach necessarily fails for even integers $p \leq k-1$. In spite of this, we actually prove the following result that generalizes Theorem~\ref{thm:main} to ``almost all'' $p \geq 1$ (including non-integer $p$).

\begin{theorem}
	\label{thm:all_p_intro}
	For any constant integer $k \geq 2$, there is an efficient reduction from $k$-SAT with $n$ variables and $m$ clauses to $\CVP_p$ on a lattice of rank $n$ (with ambient dimension $n + O(m)$) for any $p \geq 1$ such that
	\begin{enumerate}
		\item \label{item:odd} $p$ is an odd integer or $p = \infty$;
		\item \label{item:all_but_finite} $p \notin S_k$, where $S_k$ is some finite set (containing all even integers $p \leq k-1$); or
		\item \label{item:p_plus_eps} $p = p_0 + \delta(n)$ for any $p_0 \geq 1$ and any $\delta(n) \neq 0$ that converges to zero as $n \to \infty$.
	\end{enumerate}
	
	In particular, if SETH holds then for any constant $\eps > 0$, there is no $2^{(1-\eps)n}$-time algorithm for $\CVP_p$ for any $p \geq 1$ such that
	\begin{enumerate}
		\item $p$ is an odd integer or $p = \infty$;
		\item $p \notin S_k$ for some sufficiently large $k$ (depending on $\eps$); or
		\item $p = p_0 + \delta(n)$.
	\end{enumerate}
\end{theorem}

Notice that this lower bound (Theorem~\ref{thm:all_p_intro}) comes tantalizingly close to resolving the quantitative complexity of $\CVP_2$. In particular, we obtain a $2^{n}$-time lower bound on $\CVP_{2 + \delta}$ for any $0 \neq \delta(n) = o(1)$, and the fastest algorithm for $\CVP_2$ run in time $2^{n + o(n)}$. But, formally, Theorems~\ref{thm:main} and~\ref{thm:all_p_intro} say nothing about $\CVP_2$. (Indeed, there is at least some reason to believe that $\CVP_2$ is easier than $\CVP_p$ for $p \neq 2$~\cite{RR06}.) 

We note that our reductions actually work for Weighted Max-$k$-SAT for all finite $p \neq \infty$, so that our hardness result holds under a weaker assumption than SETH, namely, the corresponding hypothesis for Weighted Max-$k$-SAT. 

Finally, we note that in the special case of $p = \infty$, our reduction works even for approximate $\CVP_\infty$, or even approximate $\SVP_\infty$, with an approximation factor of $\gamma := 1+2/(k-1)$. In particular, $\gamma$ is constant for fixed $k$. This implies that for every constant $\eps > 0$, there is a constant $\gamma_\eps > 1$ such that no $2^{(1-\eps)n}$-time algorithm approximates $\SVP_\infty$ or $\CVP_\infty$ to within a factor of $\gamma_\eps$ unless SETH fails.

\paragraph{Quantitative hardness of approximate $\CVP$.  }
As we discussed above, many $2^{O(n)}$-time algorithms for $\CVP_p$ only work for $\gamma$-approximate $\CVP_p$ for constant approximation factors $\gamma > 1$. However, the reduction described above only works for \emph{exact} $\CVP_p$ (except when $p = \infty$).\footnote{One can likely show that our ``exact'' reductions actually work for $\gamma$-approximate $\CVP_p$ with some approximation factor $\gamma = 1 + o(1)$. But, this is not very interesting because standard techniques for ``boosting'' the approximation factor are useless for us. (They increase the rank far too much.)}

So, it would be preferable to show hardness for some constant approximation factor $\gamma > 1$. One way to show such a hardness result is via a fine-grained reduction from the problem of approximating Max-$k$-SAT to within a constant factor. Indeed, in the $k=2$ case, we show that such a reduction exists, so that there is no $2^{o(n)}$-time algorithm for approximating $\CVP_p$ to within some constant factor unless a $2^{o(n)}$-time algorithm exists for approximating Max-$2$-SAT. We also note that a $2^{o(n)}$-time algorithm for approximating Max-$2$-SAT to within a constant factor would imply one for Max-$3$-SAT as well. (See Proposition~\ref{prop:gap3sat-to-gap2sat}.)

We present this result informally here (without worrying about specific parameters and the exact definition of approximate Max-$2$-SAT). See Section~\ref{sec:approx_hard} for the formal statement.

\begin{theorem}
	There is an efficient reduction from approximating Max-$2$-SAT with $n$ variables and $m$ clauses to within a constant factor to approximating $\CVP_p$ to within a constant factor on a lattice of rank $n$ (with ambient dimension $n + O(m)$) for any finite $p \geq 1$.
\end{theorem}

\paragraph{Quantitative hardness of $\CVP$ with Preprocessing. }

$\CVP$ with Preprocessing ($\CVPP$) is the variant of $\CVP$ in which we are allowed arbitrary advice that depends on the lattice, but not the target vector. $\CVPP$ and its variants have potential applications in both cryptography (e.g.,~\cite{GPV08}) and cryptanalysis. And, an algorithm for $\CVPP_2$ is used as a subroutine in the celebrated Micciancio-Voulgaris algorithm for $\CVP_2$~\cite{MV13,BonifasD14}. The complexity of $\CVPP_p$ is well studied, with both hardness of approximation results~\cite{MicciancioCVPP,FeigeMicciancio,Regev03B,AlekhnovichKKV11,KhotPV12}, and efficient approximation algorithms~\cite{AharonovR05, cvpp}.

We prove the following quantitative hardness result for $\CVPP_p$. (See Section~\ref{sec:cvpp}.)

\begin{theorem}
	\label{thm:cvpp_intro}
	For any $1 \leq p < \infty$, there is no $2^{o(\sqrt{n})}$-time algorithm for $\CVPP$ unless there is a (non-uniform) $2^{o(n)}$-time algorithm for Max-$2$-SAT. In particular, no such algorithm exists unless (non-uniform) ETH fails.
\end{theorem}

\paragraph{Additional quantitative hardness results for $\CVP_p$. }

We also observe the following weaker hardness result for $\CVP_p$ for any $1\leq p < \infty$ based on different assumptions. The ETH-hardness of $\CVP_p$ was already known in folklore, and even written down by Samuel Yeom in unpublished work~\cite{priv:Vinod}. We present a slightly stronger theorem than what was previously known, showing a reduction from Max-$2$-SAT on $n$ variables to $\CVP_p$ on a lattice of rank $n$. (Prior to this work, we were only aware of reductions from $3$-SAT on $n$ variables to $\CVP_p$ on a lattice of rank $Cn$ for some very large constant $C > 1000$.)

\begin{theorem}
	\label{thm:maxcut}
	For any $1 \leq p < \infty$, there is an efficient reduction from Max-$2$-SAT with $n$ variables to $\CVP_p$ on a lattice of rank $n$ (and dimension $n +m$, where $m$ is the number of clauses). 
	
	In particular, for any constant $c > 0$, there is no $(\poly(n) \cdot 2^{c n})$-time algorithm for $\CVP_p$ unless there is a $2^{cn}$-time algorithm for Max-$2$-SAT, and there is no $2^{o(n)}$-time algorithm for $\CVP_p$ unless ETH fails.
\end{theorem}

The fastest known algorithm for the Max-$2$-SAT problem is the $\poly(n) \cdot 2^{\omega n/3}$-time algorithm due to Williams \cite{Williams05}, where $2 \leq \omega < 2.373$ is the matrix multiplication exponent~\cite{Williams12-fast_matrix,LeGall14}. This implies that a faster than $2^{\omega n/3}$-time algorithm for $\CVP_p$ (and $\CVP_2$ in particular) would yield a faster algorithm for Max-$2$-SAT.
(See, e.g.,~\cite{Woeginger08} Open Problem 4.7 and the preceding discussion.)

\subsection{Techniques}
\label{subsec:techniques}

\paragraph{Max-$2$-SAT. } We first show a straightforward reduction from Max-$2$-SAT to $\CVP_p$ for any $1 \leq p < \infty$. I.e., we prove Theorem~\ref{thm:maxcut}. This simple reduction will introduce some of the high-level ideas needed for our more difficult reductions.

Given a Max-$2$-SAT instance $\Phi$ with $n$ variables and $m$ clauses, we construct the lattice basis
\begin{equation}
	\label{eq:gadget_matrix}
B := 
\begin{pmatrix}
\bar{\Phi}\\
2\alpha I_n
\end{pmatrix}
\; ,
\end{equation}
where $\alpha > 0$ is some very large number and $\bar{\Phi} \in \R^{m \times n}$ is given by
\begin{equation}
	\label{eq:Phi_def}
\bar{\Phi}_{i,j} :=
\left\{
\begin{array}{rl}
2 & \text{ if the $i$th clause contains $x_j$,}\\
-2 & \text{ if the $i$th clause contains $\lnot x_j$,}\\
0 & \text{ otherwise}
\; .
\end{array}\right.
\end{equation}
I.e., the rows of $\bar{\Phi}$ correspond to clauses and the columns correspond to variables. Each entry encodes whether the relevant variable is included in the relevant clause unnegated, negated, or not at all, using $2$, $-2$, and $0$ respectively.
 (We assume without loss of generality that no clause contains repeated literals or a literal and its negation simultaneously.)
 The target $\vec{t} \in \R^{m + n}$ is given by
  \begin{equation}
  	\label{eq:t_def}
  \vec{t} :=  (
  t_1,
  t_2,
  \ldots,
  t_m,
  \alpha,
  \alpha, \ldots,
  \alpha)^T
  \; ,
  \end{equation}
 where 
 \begin{equation}
 	\label{eq:t_i_def}
 t_i := 3 - 2\eta_i
 \; ,
 \end{equation}
  where $\eta_i$ is the number of negated variables in the $i$th clause.
 
 Notice that the copy of $2\alpha I_n$  at the bottom of $B$ together with the sequence of $\alpha$'s in the last coordinates of $\vec{t}$ guarantee that any lattice vector $B\vec{z} $ with $\vec{z} \in \Z^n$ is at distance at least $\alpha n^{1/p}$ away from $\vec{t}$. Furthermore, if $\vec{z} \notin \{0,1\}^n$, then this distance increases to at least $\alpha (n-1+3^p)^{1/p}$. This is a standard gadget, which will allow us to ignore the case $\vec{z} \notin \{0,1\}^n$ (as long as $\alpha$ is large enough). I.e., we can view $\vec{z}$ as an assignment to the $n$ variables of $\Phi$.
 
 Now, suppose $\vec{z}$ does not satisfy the $i$th clause. Then, notice that the $i$th coordinate of $B \vec{z}$ will be exactly $-2\eta_i$, so that $(B\vec{z} - \vec{t})_i = 0-3 = -3$. If, on the other hand, exactly one literal in the $i$th clause is satisfied, then the $i$th coordinate of $B\vec{z}$ will be $2-2\eta_i$, so that $(B\vec{z} - \vec{t})_i = 2-3 = -1$.  Finally, if both literals are satisfied, then the $i$th coordinate will be $4-2\eta_i$, so that $(B\vec{z} - \vec{t})_i = 4-3 = 1$. In particular, if the $i$th clause is not satisfied, then $|(B\vec{z} - \vec{t})_i| = 3$. Otherwise, $|(B\vec{z} - \vec{t})_i| = 1$.
 
 It follows that the distance to the target is exactly $\dist_p(\vec{t},\lat)^p = \alpha^p n + S + 3^p(m-S) = \alpha^p n -(3^p - 1)S + 3^pm$, where $S$ is the maximal number of satisfied clauses of $\Phi$.  So, the distance $\dist_p(\vec{t}, \lat)$ tells us exactly the maximum number of satisfiable clauses, which is what we needed.

\paragraph{Difficulties extending this to $k$-SAT. } The above reduction relied on one very important fact: that $|4-3| = |2-3| < |0-3|$. In particular, a $2$-SAT clause can be satisfied in two different ways; either one variable is satisfied or two variables are satisfied. We designed our $\CVP$ instance above so that the $i$th coordinate of $B\vec{z} - \vec{t}$ is $4-3$ if two literals in the $i$th clause are satisfied by $\vec{z} \in \{0,1\}^n$, $2-3$ if one literal is satisfied, and $0-3$ if the clause is unsatisfied. Since $|4-3| = |2-3|$, the ``contribution'' of this $i$th coordinate to the distance $\|B\vec{z} - \vec{t}\|_p^p$ is the same for any satisfied clause. Since $|0-3| > |4-3|$, the contribution to the $i$th coordinate is larger for unsatisfied clauses than satisfied clauses.

Suppose we tried the same construction for a $k$-SAT instance. I.e., suppose we take $\bar{\Phi} \in \R^{m \times n}$ to encode the literals in each clause as in Eq.~\eqref{eq:Phi_def} and construct our lattice basis $B$ as in Eq.~\eqref{eq:gadget_matrix} and target $\vec{t}$ as in Eq.~\eqref{eq:t_def}, perhaps with the number $3$ in the definition of $\vec{t}$ replaced by an arbitrary $t^* \in \R$. Then, the $i$th coordinate of $B\vec{z} - \vec{t}$ would be $2S_i - t^*$, where $S_i$ is the number of literals satisfied in the $i$th clause. 

No matter how cleverly we choose $t^* \in \R$, some satisfied clauses will contribute more to the distance than others as long as $k \geq 3$. I.e., there will always be some ``imbalance'' in this contribution. 
As a result, we will not be able to distinguish between, e.g., an assignment that satisfies all clauses but has $S_i $ far from $t^*/2$ for all $i$ and an assignment that satisfies fewer clauses but has $S_i \approx t^*/2$ whenever $i$ corresponds to a satisfying clause.

In short, for $k \geq 3$, we run into trouble because satisfying assignments to a clause may satisfy anywhere between $1$ and $k$ literals, but $k$ distinct numbers obviously cannot all be equidistant from some number $t^*$. (See Section~\ref{sec:ETH-MAX-CUT} for a simple way to get around this issue by adding to the rank of the lattice. Below, we show a more technical way to do this without adding to the rank of the lattice, which allows us to prove SETH-hardness.)

\paragraph{A solution via isolating parallelepipeds. } To get around the issue described above for $k \geq 3$, we first observe that, while many distinct \emph{numbers} cannot all be equidistant from some \emph{number} $t^*$, it is trivial to find many distinct \emph{vectors} in $\R^{d^*}$ that are equidistant from some \emph{vector} $\vec{t}^* \in \R^{d^*}$.

We therefore consider modifying the reduction from above by replacing the scalar $\pm 2$ values in our matrix $\bar{\Phi}$ with vectors in $\R^{d^*}$ for some $d^*$. In particular, for some vectors $V = (\vec{v}_1,\ldots, \vec{v}_k) \in \R^{d^* \times k}$, we define $\bar{\Phi} \in \R^{d^* m \times n}$ as 
\begin{equation}
	\label{eq:Phi_k_def}
	\bar{\Phi}_{i,j} :=
		\left\{
		\begin{array}{rl}
		\vec{v}_s & \text{ if $x_j$ is the $s$th literal in the $i$th clause},\\
		-\vec{v}_s & \text{ if $\lnot x_j$ is the $s$th literal in the $i$th clause},\\
		\vec{0}_d & \text{ otherwise}
		\; ,
	\end{array}\right.
\end{equation}
where we have abused notation and taken $\bar{\Phi}_{i,j}$ to be a column vector in $d^*$ dimensions.
By defining $\vec{t} \in \R^{d^* m + n}$ appropriately,%
\footnote{In particular, we replace the scalars $t_i$ in Eq.~\eqref{eq:t_i_def} with vectors 
	\[\vec{t}_i := \vec{t}^*- \sum \vec{v}_s \in \R^{d^*}
	\; ,
	\] 
where the sum is over $s$ such that the $s$th literal in the $i$th clause is negated.
} 
we will get that the ``contribution of the $i$th clause to the distance'' $\|B\vec{z} - \vec{t}\|_p^p$ is exactly $\|V \vec{y} - \vec{t}^*\|_p^p$ for some $\vec{t}^* \in \R^{d^*}$, where $\vec{y} \in \{0,1\}^k$ such that $y_s = 1$ if and only if $\vec{z}$ satisfies the $s$th literal of the relevant clause. (See Table~\ref{tbl:piped_reduction} for a diagram showing the output of the reduction and Theorem~\ref{thm:isolating_implies_reduction} for the formal statement.) We stress that, while we have increased the \emph{ambient dimension} by nearly a factor of $d^*$, the \emph{rank} of the lattice is still $n$.

This motivates the introduction of our primary technical tool, which we call \emph{isolating parallelepipeds}. For $1\leq p \leq \infty$, a $(p,k)$-isolating parallelepiped is represented by a matrix $V \in \R^{d^* \times k}$ and a shift vector $\vec{t}^* \in \R^{d^*}$ with the special property that one vertex of the parallelepiped $V \{ 0,1\}^k - \vec{t}^*$ is ``isolated.'' (Here, $V \{ 0,1\}^k - \vec{t}^*$ is an affine transformation of the hypercube, i.e., a parallelepiped.) In particular, every vertex of the parallelepiped, $V \vec{y} - \vec{t}^*$ for $\vec{y} \in \{0,1\}^k$ has unit length $\|V \vec{y} - \vec{t}^*\|_p = 1$ \emph{except for the vertex $-\vec{t}^*$}, which is longer, i.e., $\|\vec{t}^*\|_p > 1$. (See Figure~\ref{fig:l2-parallelpiped}.)

In terms of the reduction above, an isolating parallelepiped is exactly what we need. In particular, if we plug $V$ and $\vec{t}^*$ into the above reduction, then all satisfied clauses (which correspond to non-zero $\vec{y}$ in the above description) will ``contribute'' $1$ to the distance $\|B\vec{z} - \vec{t}\|_p^p$, while unsatisfied clauses (which correspond to $\vec{y} = \vec{0}$) will contribute $1+\delta$ for some $\delta > 0$. Therefore, the total distance will be exactly $\|B\vec{z} - \vec{t}\|_p^p = \alpha^p n + m^+(\vec{z}) + (m-m^+(\vec{z}))(1+\delta) = \alpha^p n - \delta m^+(\vec{z}) + (1+\delta)m$, where $m^+(\vec{z})$ is the number of clauses satisfied by $\vec{z}$. So, the distance $\dist_p(\vec{t}, \lat)$ exactly corresponds to the maximal number of satisfied clauses, as needed.

\begin{figure}
	\begin{center}
		\begin{tabular}{m{4cm} m{0cm} m{4cm}}
			\hspace{-20ex} \includegraphics{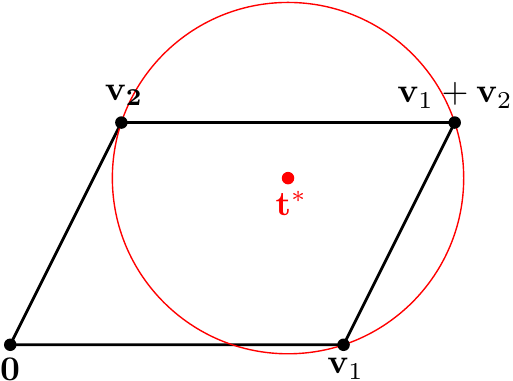} &  &  \includegraphics[scale=0.8]{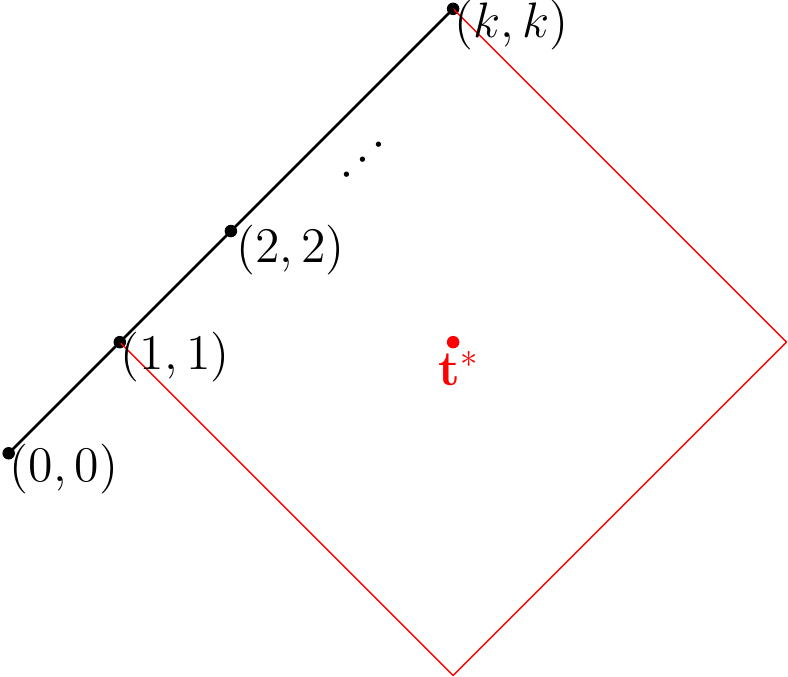}
		\end{tabular}
	\end{center}
	\caption{$(p, k)$-isolating parallelepipeds for $p = 2, k = 2$ (left) and $p = 1, k \geq 1$ (right). On the left, the vectors $\vec{v}_1$, $\vec{v}_2$, and $\vec{v}_1 + \vec{v}_2$ are all at the same distance from $\vec{t}^*$, while $\vec{0}$ is strictly farther away. 
		On the right is the degenerate parallelepiped generated by $k$ copies of the vector $(1,1)$. The vectors $(i, i)$ are all at the same $\ell_1$ distance from $\vec{t}^*$ for $1 \leq i \leq m$, while $(0, 0)$ is strictly farther away. 
		The (scaled) unit balls centered at $\vec{t}^*$ are shown in red, while the parallelepipeds are shown in black.}
	\label{fig:l2-parallelpiped}
\end{figure}

\paragraph{Constructing isolating parallelepipeds. } Of course, in order for the above to be useful, we must show how to construct these $(p,k)$-isolating parallelepipeds. Indeed, it is not hard to find constructions for all $p \geq 1$ when $k = 2$, and even for all $k$ in the special case when $p = 1$ (see Figure~\ref{fig:l2-parallelpiped}). Some other fairly nice examples can also be found for small $k$, as shown in Figure~\ref{fig:ell_3_k=3}. For $p > 1$ and large $k$, these objects seem to be much harder to find. (In fact, in Section~\ref{sec:limitations}, we show that there is no $(p,k)$-isolating parallelepiped for any even integer $p \leq k-1$.) Our solution is therefore a bit technical.

\begin{figure}
	\[
	V := \frac{1}{2 \cdot 12^{1/3}} \cdot 
	\left(
	\begin{array}{rrr}
	12^{1/3} & 12^{1/3} & 12^{1/3} \\
	1 & 1 & -1 \\
	1 & -1 & 1\\
	1 & -1 & -1\\
	-1 & 1 & 1\\
	-1 & 1 & -1\\
	-1 & -1 & 1
	\end{array}
	\right),
	\qquad
	\vec{t}^* := \frac{1}{2 \cdot 12^{1/3}} \cdot
	\begin{pmatrix}
	2 \cdot 12^{1/3}\\
	2 \\
	2\\
	2\\
	2\\
	2\\
	2
	\end{pmatrix}
	\; .
	\]
	\caption{\label{fig:ell_3_k=3} A $(3,3)$-isolating parallelepiped in seven dimensions. One can verify that $\|V \vec{y} - \vec{t}^*\|_3^3 = 1$ for all non-zero $\vec{y} \in \{0,1\}^3$, and $\|\vec{t}^*\|_3^3 = 3/2$.}
\end{figure}

At a high level, in Section~\ref{sec:SETH}, we consider a natural class of parallelepipeds $V \in \R^{2^k \times k}, \vec{t}^* \in \R^{2^k}$ parametrized by some weights $\alpha_0,\alpha_1,\ldots, \alpha_k \geq 0$ and a scalar shift $t^* \in \R$. These parallelepipeds are constructed so that the length of the vertex $\|V \vec{y} - \vec{t}^*\|_p^p$ for $\vec{y} \in \{0,1\}^k$ depends only on the Hamming weight of $\vec{y}$ and is linear in the $\alpha_i$ for fixed $t^*$. In other words, there is a matrix $M_k(p,t^*) \in \R^{(k+1)\times (k+1)}$ such that $M_k(p,t^*) (\alpha_0,\ldots, \alpha_k)^T $ encodes the value of $\|V \vec{y} - \vec{t}^*\|_p^p$ for each possible Hamming weight of $\vec{y} \in \{0,1\}^k$. (See Lemma~\ref{lem:alpha_norm_thing}.)

We show that, in order to find weights $\alpha_0,\ldots, \alpha_k \geq 0$ such that $V$ and $\vec{t}^*$ define a $(p,k)$-isolating parallelepiped, it suffices to find a $t^*$ such that $M_k(p,t^*)$ is invertible. For each odd integer $p \geq 1$ and each $k \geq 2$, we show an algorithm that finds such a $t^*$. (See Section~\ref{sec:odd_p}.) 

To extend this result to other $p \geq 1$, we consider the determinant of $M_k(p,t^*)$ for fixed $k$ and $t^*$, viewed as a function of $p$. We observe that this function has a rather nice form---it is a Dirichlet polynomial. I.e., for fixed $t^*$ and $k$, the determinant can be written as $\sum \exp(a_i p)$ for some $a_i \in \R$. Such a function has finitely many roots unless it is identically zero. So, we take the value of $t^*$ from above such that, say, $M_k(1,t^*)$ is invertible. Since $M_k(1,t^*)$ does not have zero determinant, the Dirichlet polynomial corresponding to $\det(M_k(p,t^*))$ cannot be identically zero and therefore has finitely many roots. This is how we prove Theorem~\ref{thm:all_p_intro}. (See Section~\ref{sec:all_p}.)

\paragraph{Extension to constant-factor approximation. }
In order to extend our hardness results to approximate $\CVP_p$ for finite $p$, we can try simply using the same reduction with $k$-SAT replaced by approximate Max-$k$-SAT. Unfortunately, this does not quite work. Indeed, it is easy to see that the ``identity matrix gadget'' that we use to restrict our attention to lattice vectors whose coordinates are in $\{0,1\}$ (Eq.~\eqref{eq:gadget_matrix}) cannot tolerate an approximation factor larger than $1+O(1/n)$ (for finite $p$). 

However, we observe that when $k = 2$, this identity matrix gadget is actually unnecessary. In particular, even without this gadget, it ``never helps'' to consider a lattice vector whose coordinates are not all in $\{0,1\}$. It then follows immediately from the analysis above that Gap-$2$-SAT reduces to approximate $\CVP_p$ with a constant approximation factor strictly greater than one. 
We note that we do not know how to extend this result to larger $k > 2$ (except when $p = 1$, see Theorem~\ref{thm:cvp1-hardness-of-approx}).
We show that the case $k = 2$ is sufficient for proving Gap-ETH-hardness (see Proposition~\ref{prop:gap3sat-to-gap2sat}), but we suspect that one can   just ``remove the identity matrix gadget'' from all of our reductions for finite $p$. If this were true, it would show Gap-ETH-hardness of approximation for slightly larger constant approximation factors and imply even stronger hardness results under less common assumptions.

\subsection{Open questions}

The most important question that we leave open is the extension of our SETH-hardness result to arbitrary $p \geq 1$. In particular, while our result applies to $p = p(n) \neq 2$ that approaches $2$ asymptotically, it does not apply to the specific case $p = 2$. An extension to $p = 2$ would settle the time complexity of $\CVP_2$ up to a factor of $2^{o(n)}$ (assuming SETH). However, we know that our technique does not work in this case (in that $(2,k)$-parallelepipeds do not exist for $k \geq 3$), so substantial new ideas might be needed to resolve this issue.

Another direction would be to strengthen our hardness of approximation results in one of two possible directions. First, one could try to increase the approximation factor. (Prior techniques for amplifying the approximation factor increase the rank of the lattice quite a bit, so they do not yield very interesting quantitative hardness results.)  Second, one could try to show a reduction from Gap-$k$-SAT to approximate $\CVP_p$ for $k \geq 3$. 
For $p \in \{1,\infty\}$, we already have such a reduction, and as we mentioned above, we suspect that we can simply ``remove the identity matrix gadget'' in our current reduction to achieve this for $1 < p < \infty$. But, we do not know how to prove that this works.

Finally, we note that our main reduction constructs lattices of \emph{rank} $n$, but the ambient dimension $d$ can be significantly larger. (Specifically, $d = n+ O(m)$, where $m$ is the number of clauses in the relevant SAT instance, and where the hidden constant depends on $k$ and can be very large.) Lattice problems are typically parameterized in terms of the rank of the lattice (and for the $\ell_2$ norm, one can assume without loss of generality that $d = n$), but it is still interesting to ask whether we can reduce the ambient dimension $d$.

\subsection*{Organization}

In Section~\ref{sec:prelims}, we review some necessary background knowledge. In Section~\ref{sec:geom-char}, we show how to use a $(p,k)$-isolating parallelepiped (for finite $p$) to reduce any $n$-variable instance of $k$-SAT to a $\CVP_p$ instance with rank $n$, and we show that this immediately gives SETH-hardness for $p = 1$. In Section~\ref{sec:SETH}, we show how to construct $(p,k)$-isolating parallelepipeds, first for odd integers $p \geq 1$ and then for ``almost all'' $p$. In Section~\ref{sec:approx_hard}, we show $2^{\Omega(n)}$-hardness of approximating $\CVP_p$ up to a constant factor. In Section~\ref{sec:other}, we prove a number of additional hardness results: $2^{\Omega(\sqrt{n})}$ ETH- and Max-$2$-SAT-hardness of $\CVPP_p$ (Section~\ref{sec:cvpp}), ETH- and Max-$2$-SAT-hardness of $\CVP_p$ (Section~\ref{sec:ETH-MAX-CUT}), 
and SETH-hardness of $\CVP_\infty$ and $\SVP_\infty$ (Section~\ref{sec:infinity!}). 

\section{Preliminaries}
\label{sec:prelims}
Throughout this paper, we work with lattice problems over $\R^d$ for convenience. Formally, we must pick a suitable representation of real numbers and consider both the size of the representation and the efficiency of arithmetic operations in the given representation. But, we omit such details throughout to ease readability. 

\subsection{Computational lattice problems}
\label{sec:problem_defs}

\renewcommand{\size}{\mathrm{size}}

Let $\dist_p(\lat, \vec{t}) := \min_{\vec{x} \in \lat} \norm{\vec{x} - \vec{t}}_p$ denote the $\ell_p$ distance of $\vec{t}$ to $\lat$. 
In addition to SVP and CVP, we also consider a variant of CVP called the Closest Vector Problem with Preprocessing (CVPP), which allows arbitrary preprocessing of a lattice.

\begin{definition}
For any $\gamma \geq 1$ and $1 \leq p \leq \infty$, the $\gamma$-approximate Shortest Vector Problem with respect to the $\ell_p$ norm ($\gamma$-$\SVP_p$) is the promise problem defined as follows. Given a lattice $\lat$ (specified by a basis $B \in \R^{d \times n}$) and a number $r > 0$, distinguish between a `YES' instance where there exists a non-zero vector $\vec{v} \in \lat$ such that $\norm{\vec{v}}_p \leq r$, and a `NO' instance where $\norm{\vec{v}}_p > \gamma r$ for all non-zero $v \in \lat$.
\end{definition}

\begin{definition}
For any $\gamma \geq 1$ and $1 \leq p \leq \infty$, the $\gamma$-approximate Closest Vector Problem with respect to the $\ell_p$ norm ($\gamma$-$\CVP_p$) is the promise problem defined as follows. Given a lattice $\lat$ (specified by a basis $B \in \R^{d \times n}$), a target vector $\vec{t} \in \R^d$, and a number $r > 0$, distinguish between a `YES' instance where $\dist_p(\lat, \vec{t}) \leq r$, and a `NO' instance where $\dist_p(\lat, \vec{t}) > \gamma r$.
\end{definition}

\noindent When $\gamma = 1$, we refer to the problems simply as $\SVP_p$ and $\CVP_p$, respectively.

\begin{definition}
The Closest Vector Problem with Preprocessing with respect to the $\ell_p$ norm ($\CVPP_p)$ is the problem of finding a preprocessing function $P$ and an algorithm $Q$ which work as follows. Given a lattice $\lat$ (specified by a basis $B \in \R^{d \times n}$), $P$ outputs a new description of $\lat$. Given $P(\lat)$, a target vector $\vec{t} \in \R^d$, and a number $r > 0$, $Q$ decides whether $\dist_p(\lat, \vec{t}) \leq r$.
\end{definition}

When we measure the running time of a $\CVPP$ algorithm, we only count the running time of $Q$, and not of the preprocessing algorithm $P$.

\subsection{Satisfiability problems and the Max-Cut problem}
A $k$-SAT formula $\Phi$ on $n$ Boolean variables $x_1, \ldots, x_n$ and $m$ clauses $C_1, \ldots, C_m$ is a conjunction $\Phi = \land_{i=1}^m C_i$ of clauses $C_i = \lor_{s=1}^k \ell_{i,s}$, where the literals $\ell_{i,s}$ denote a variable $x_j$ or its negation $\lnot x_j$. The goal is to decide whether there exists an assignment $\vec{a} \in \set{0, 1}^n$ to the variables of $\Phi$ such that all clauses have at least one ``true'' literal, i.e., so that all clauses are satisfied.

The \emph{value} of a $k$-SAT formula $\Phi$, denoted $\val(\Phi)$, is the maximum fraction of clauses satisfied by an assignment to $\Phi$.

\begin{definition}
Given a $k$-SAT formula $\Phi$ and constants $0 \leq \delta < \eps \leq 1$, the $(\delta, \eps)$-Gap-$k$-SAT problem is the promise problem defined as follows. The goal is to distinguish between a `YES' instance in which $\val(\Phi) \geq \eps$, and a `NO' instance in which $\val(\Phi) < \delta$.
\end{definition}

\begin{definition}
Given a $k$-SAT formula $\Phi$ with clauses $\mathcal{C} = \set{C_1, \ldots, C_m}$, a clause weight function $w: \mathcal{C} \to \Z^+$, and a weight threshold $W$, the Weighted Max-$k$-SAT problem is to decide whether there exists an assignment $\vec{a}$ to the variables of $\Phi$ such that $\sum_{\textrm{$C_i$ is sat. by $\vec{a}$}} w(C_i) \geq W$.
\end{definition}

\begin{definition}
Given an undirected graph $G = (V, E)$, an edge weight function $w: E \to \Z^+$, and a weight threshold $W$, the Weighted Max-CUT problem is defined as follows. The goal is to decide whether $V$ can be partitioned into sets $V_1$ and $V_2$ such that $\sum_{e_{ij} \in E : v_i \in V_1, v_j \in V_2} w(e_{ij}) \geq W$.
\end{definition}

There exists a folklore reduction from an instance of Weighted Max-Cut on a graph with $n$ vertices to an instance of Weighted Max 2-SAT on a formula with $n$ variables. See, e.g.,~\cite{journals/dam/GrammHNR03}.

\subsection{Exponential Time Hypotheses}
\label{sec:prelims_SETH}

Impagliazzo and Paturi~\cite{IP1999} introduced the following two hypotheses (ETH and SETH), which are now widely used to study the quantitative hardness of computational problems.

\begin{definition}
The Exponential Time Hypothesis (ETH) is the hypothesis defined as follows. For every $k \geq 3$ there exists a constant $\eps > 0$ such that no algorithm solves $k$-SAT formulas with $n$ variables in $2^{\eps n}$-time. In particular, there is no $2^{o(n)}$-time algorithm for 3-SAT.
\label{def:eth}
\end{definition}

\begin{definition}
The Strong Exponential Time Hypothesis (SETH) is the hypothesis defined as follows. For every constant $\eps > 0$ there exists $k$ such that no algorithm solves $k$-SAT formulas with $n$ variables in $2^{(1 - \eps) n}$-time.
\label{def:seth}
\end{definition}

An important tool in the study of the exact complexity of $k$-SAT is the Sparisification Lemma of Impagliazzo, Paturi, and Zane~\cite{journals/jcss/ImpagliazzoPZ01} which roughly says that any $k$-SAT formula can be replaced with $2^{\eps n}$ formulas each with $O(n)$ clauses for some $\eps > 0$.

\begin{proposition}[Sparsification Lemma,~\cite{journals/jcss/ImpagliazzoPZ01}]
For every $k \in \Z^+$ and $\eps > 0$ there exists a constant $c = c(k, \eps)$ such that any $k$-SAT formula $\Phi$ with $n$ variables can be expressed as $\Phi = \lor_{i=1}^r \Psi_{i}$ where $r \leq 2^{\eps n}$ and each $\Psi_i$ is a $k$-SAT formula with at most $cn$ clauses. Moreover, this disjunction can be computed in $2^{\eps n}$-time.
\label{prop:sparsification}
\end{proposition}

In this paper we also consider the W-Max-SAT-SETH hypothesis, which corresponds to SETH but with Weighted Max-$k$-SAT in place of $k$-SAT.
Our main result only relies on this weaker variant of SETH, and is therefore more robust.

Dinur~\cite{journals/eccc/Dinur16} and Manurangsi and Raghavendra~\cite{conf/icalp/ManurangsiR17} recently introduced a ``gap'' version of ETH, which asserts that Gap-$3$-SAT takes $2^{\Omega(n)}$-time.

\begin{definition}
The (randomized) Gap-Exponential Time Hypothesis ((randomized) Gap-ETH) is the hypothesis that there exist constants $\delta < 1$ and $\eps > 0$ such that no (randomized) algorithm solves $(\delta, 1)$-Gap-$3$-SAT instances with $n$ variables
in $2^{\eps n}$-time.
\label{def:gap-eth}
\end{definition}

As Dinur~\cite{journals/eccc/Dinur16} notes, one can sparsify a Gap-SAT instance simply by sampling clauses. Therefore, we can assume (almost) without loss of generality that Gap-ETH applies only to formulas with $O(n)$ clauses. The caveat is that the sampling is randomized, so finding a $2^{o(n)}$-time algorithm for sparse Gap-$3$-SAT only implies a randomized $2^{o(n)}$-time algorithm for general Gap-$3$-SAT. %

We give a variant of Dinur's sampling argument in Proposition~\ref{prop:gap-sat-sparsification}. The idea is to show that both the total number of sampled clauses and the number of sampled clauses that are satisfied by any given assignment are highly concentrated around their expectation by using the Chernoff bound, and then to take a union bound over the bad events where these quantities deviate substantially from their expectation.

We will use the following multiplicative Chernoff bounds (see, e.g.,~\cite{hp/chernoff}). Let $X_1, \ldots, X_n$ be independent identically distributed Bernoulli random variables with expectation $p$, so that the expectation of $\sum_{i=1}^n X_i$ is $\mu = pn$. Then:

\begin{equation}
\Pr[\sum_{i=1}^n X_i < (1 - \alpha)\mu] < e^{-\mu \alpha^2/2}  \ ,
\label{eq:mult-chernoff-lb}
\end{equation}

\begin{equation}
\Pr[\sum_{i=1}^n X_i > (1 + \alpha)\mu] < e^{-\mu \alpha^2/4} \ .
\label{eq:mult-chernoff-ub}
\end{equation}

\begin{proposition}[Sparsification for Gap-SAT]
For any $0 < \delta < \delta' < 1$, there is a polynomial-time randomized reduction from a $(\delta, 1)$-Gap-$k$-SAT instance $\Phi$ with $n$ variables and $m$ clauses to a $(\delta', 1)$-Gap-$k$-SAT instance $\Phi'$ with $n$ variables and $O(n)$ clauses.
\label{prop:gap-sat-sparsification}
\end{proposition}

\begin{proof}
Let $\Phi'$ be the formula obtained by sampling each clause of $\Phi$ independently with probability $p := \min \set{1, 10/(\delta \alpha^2) \cdot n/m}$, where $\alpha$ is fixed so that $-(1 - \delta'/\delta)/(1 + \delta'/\delta) < \alpha < 1$. Clearly, if $\val(\Phi) = 1$ then $\val(\Phi') = 1$ as well. We analyze the case where $\val(\Phi) < \delta$.

In expectation $\Phi'$ has $pm$ clauses. Furthermore, because $\val(\Phi') < \delta$, in expectation any fixed assignment will satisfy fewer than $\delta pm$ clauses of $\Phi'$.
Therefore by Equation~\eqref{eq:mult-chernoff-lb},
\begin{equation}
\Pr[\textrm{Number of clauses in $\Phi'$} \leq (1 - \alpha)pm] \leq e^{-\alpha^2 pm/2} \leq e^{-2n}.
\label{eq:sparsification-num-clauses}
\end{equation}
Furthermore, by Equation~\eqref{eq:mult-chernoff-ub}, we have that for each fixed assignment $\vec{a}$,
\begin{equation}
\Pr[\textrm{Number of clauses in $\Phi'$ sat. by $\vec{a}$} \geq (1 + \alpha)\delta pm] \leq e^{-\alpha^2 \delta pm /4} \leq e^{-2n}.
\label{eq:sparsification-sat-clauses}
\end{equation}

By applying Equations~\eqref{eq:sparsification-num-clauses} and~\eqref{eq:sparsification-sat-clauses}, and taking a union bound we get that the probability that $\Phi'$ has at least $(1 - \alpha)pm$ clauses and that no assignment to $\Phi'$ satisfies more than $(1 + \alpha)\delta pm$ clauses is at least $1 - (e^{-2n} + 2^n e^{-2n}) \geq 1 - 2e^{-n}$. Therefore,
\[
\val(\Phi') \leq \frac{(1 + \alpha) pm}{(1 - \alpha)pm} \cdot \delta < \delta'
\]
with high probability.
\end{proof}

Additionally, we will use a reduction of Garey et al.~\cite{GJS76} from $3$-SAT to Max-$2$-SAT which also works as a reduction from Gap-$3$-SAT to Gap-$2$-SAT. The reduction works by outputting ten $1$- and $2$-clauses for each $3$-clause in the original formula. Any assignment which satisfies the original clause corresponds to an assignment which satisfies $7$ of the output clauses, and any assignment which does not satisfy the original clause corresponds to an assignment which satisfies $6$ of the output clauses.

\begin{proposition}[{\cite[Theorem 1.1]{GJS76}}]
For every $0 \leq \delta < \eps \leq 1$, there is a polynomial-time reduction from every instance of $(\delta, \eps)$-Gap-$3$-SAT with $n$ variables and $m$ clauses to an instance of $((6 + \delta)/10, (6 + \eps)/10)$-Gap-$2$-SAT with $n + m$ variables and $10m$ clauses.
\label{prop:gap3sat-to-gap2sat}
\end{proposition}

\section{SETH-hardness from isolating parallelepipeds}
\label{sec:geom-char}

\setlength{\tabcolsep}{0.3cm}
\begin{table}[t]
	\begin{center}
		\begin{tabular}{cc}
			\begin{tabular}{c|ccccc}
				& $x_1$ & $x_2$ & $\cdots$ & $x_{n-1}$ & $x_n$ \\ \hline
				\multirow{3}{*}{$C_1$\Bigg\{} & \multirow{3}{*}{$\vec{v}_1$} & \multirow{3}{*}{$\vec{v}_2$} & \multirow{3}{*}{$\cdots$} & \multirow{3}{*}{$\vec{0}_{d^*}$} & \multirow{3}{*}{$-\vec{v}_3$} \\
				& &  &  &  &  \\
				&  &  &  &  &  \\
				\vdots & $\vdots$ & $\cdots$ & $\ddots$ & $\vdots$ & $\vdots$\\
				\multirow{3}{*}{$C_m\Bigg\{$} & \multirow{3}{*}{$\vec{0}_{d^*}$} & \multirow{3}{*}{$-\vec{v}_1$} & \multirow{3}{*}{$\cdots$} & \multirow{3}{*}{$\vec{v}_2$} & \multirow{3}{*}{$\vec{v}_3$}\\
				&  &  &  &  & \\
				&  &  &  &  & \\ \hline
				$x_1$ & $2 \alpha^{1/p}$ & $0$ & $\cdots$ & $0$ & $0$\\
				$x_2$ & $0$ & $2 \alpha^{1/p}$ & $\cdots$ & $0$ & $0$\\
				$\vdots$ & $\vdots$ & $\vdots$ & $\ddots$ & $\vdots$ & $\vdots$\\
				$x_{n-1}$ & $0$ & $0$ & $\cdots$ & $2 \alpha^{1/p}$ & $0$\\
				$x_n$ & $0$ & $0$ & $\cdots$ & $0$ & $2 \alpha^{1/p}$\\ \hline
			\end{tabular} &
			\begin{tabular}{c}
				\\ \hline
				\multirow{3}{*}{$\vec{t}^*-\vec{v}_3$} \\
				 \\
				 \\
				$\vdots$ \\ 
				\multirow{3}{*}{$\vec{t}^*-\vec{v}_1$} \\
				\\
				\\ \hline
				$\alpha^{1/p}$ \\
				$\alpha^{1/p}$ \\
				$\vdots$\\
				$\alpha^{1/p}$\\
				$\alpha^{1/p}$ \\ \hline
			\end{tabular} \\ \\
			$B$ & $\vec{t}$
		\end{tabular}
	\end{center}
	\caption{A basis $B$ and target vector $\vec{t}$ output by the reduction from Theorem~\ref{thm:isolating_implies_reduction} with some $(p,3)$-isolating parallelepiped given by $V = (\vec{v}_1,\vec{v}_2, \vec{v}_3) \in \R^{d^* \times 3}$ and $\vec{t}^* \in \R^{d^*}$. In this example, the first clause is $C_1 \equiv x_1 \lor x_2 \lor \lnot x_n$ and the $m$th clause is $C_m \equiv \lnot x_2 \lor x_{n-1} \lor x_n$. By the definition of an isolating parallelepiped (Definition~\ref{def:piped}), the contribution of the first $d$ coordinates to the distance $\|B\vec{z} - \vec{t}\|_p^p$ will be $1$ for any assignment $\vec{z} \in \{0,1\}^n$ satisfying $C_1$, while non-satisfying assignments contribute $(1+\delta)$ for some $\delta > 0$. For example, if $z_1=1,z_2=0,z_n=1$, the clause $C_1$ is satisfied, and the first $d$ coordinates will contribute $\|\vec{v}_1-\vec{v}_3-(\vec{t}^*-\vec{v}_3)\|_p^p=\|\vec{v}_1-\vec{t}^*\|_p^p=1$. On the other hand, if $z_1=0,z_2=0,z_n=1$, then $C_1$ is not satisfied, and $\|-\vec{v}_3-(\vec{t}^*-\vec{v}_3)\|_p^p=\|\vec{t}^*\|_p^p=1+\delta$.}
	\label{tbl:linf-red}
	\label{tbl:piped_reduction}
\end{table}

We start by giving a reduction from instances of weighted Max-$k$-SAT on formulas with $n$ variables to instances of $\CVP_p$ with rank $n$ for all $p$ that uses a certain geometric object, which we define next. Let $\vec{1}_n$ and $\vec{0}_n$ denote the all 1s and all 0s vectors of length $n$ respectively, and let $I_n$ denote the $n \times n$ identity matrix.

\begin{definition}
	\label{def:piped}
For any $1 \leq p \leq \infty$ and integer $k \geq 2$, we say that $V \in \R^{d^* \times k}$ and $\vec{t}^* \in \R^{d^*}$ define a $(p, k)$-\emph{isolating parallelepiped} if $\norm{\vec{t}}_p > 1$ and $\norm{V\vec{x} - \vec{t}^*}_p = 1$ for all $\vec{x} \in \set{0, 1}^k \setminus \set{\vec{0}_k}$.
\end{definition}

In order to give the reduction, we first introduce some notation related to SAT. Let $\Phi$ be a $k$-SAT formula on $n$ variables $x_1, \ldots, x_n$ and $m$ clauses $C_1, \ldots, C_m$. Let $\ind(\ell)$ denote the index of the variable underlying a literal $\ell$. I.e., $\ind(\ell) = j$ if $\ell = x_j$ or $\ell = \lnot x_j$. Call a literal $\ell$ \emph{positive} if $\ell = x_j$ and \emph{negative} if $\ell = \lnot x_j$ for some variable $x_j$.
Given a clause $C_i = \lor_{s=1}^k \ell_{i, s}$, let $P_i := \set{s \in [k] : \ell_{i, s} \textrm{ is positive}}$ and let $N_i := \set{s \in [k] : \ell_{i, s} \textrm{ is negative}}$ denote the indices of positive and negative literals in $C_i$ respectively. Given an assignment $\vec{a} \in \{0,1\}^n$ to the variables of $\Phi$, let $S_i(\vec{a})$ denote the indices of literals in $C_i$ satisfied by $\vec{a}$. I.e., $S_i(\vec{a}) := \set{s \in P_i : a_{\ind(\ell_{i,s})} = 1} \cup \set{s \in N_i : a_{\ind(\ell_{i,s})} = 0}$. Finally, let $m^+(\vec{a})$ denote the number of clauses of $\Phi$ satisfied by the assignment $\vec{a}$, i.e., the number of clauses $i$ for which $|S_i(\vec{a})| \geq 1$.

\begin{theorem}
	\label{thm:isolating_implies_reduction}
If there exists a computable $(p, k)$-isolating parallelepiped for some $p = p(n) \in [1, \infty)$ and integer $k \geq 2$, then there exists a polynomial-time reduction from any (weighted-)Max-$k$-SAT instance with $n$ variables to a $\CVP_p$ instance of rank $n$.
\label{thm:max-k-sat-to-low-rank}
\end{theorem}

\begin{proof}
For simplicity, we give a reduction from unweighted Max-$k$-SAT, and afterwards sketch how to modify our reduction to handle the weighted case as well.
Namely, we give a reduction from any Max-$k$-SAT instance $(\Phi, W)$ to an instance $(B, \vec{t}^*, r)$ of $\CVP_p$. Here, the formula $\Phi$ is on $n$ variables $x_1, \ldots, x_n$ and $m$ clauses $C_1, \ldots, C_m$. $(\Phi, W)$ is a `YES' instance if there exists an assignment $\vec{a}$ such that $m^+(\vec{a}) \geq W$.

By assumption, there exist computable $d^* = d^*(p, k) \in \Z^+$, $V = [\vec{v}_1, \ldots, \vec{v}_k] \in \R^{d^* \times k}$, and $\vec{t}^* \in \R^{d^*}$ such that $\norm{\vec{t}^*}_p = (1 + \delta)^{1/p}$ for some $\delta > 0$ and $\norm{V\vec{z} - \vec{t}^*}_p = 1$ for all $\vec{z} \in \set{0, 1}^k \setminus \set{\vec{0}_k}$.

We define the output $\CVP_p$ instance as follows. Let $d := md^* + n$. The basis $B \in \R^{d \times n}$ and target vector $\vec{t} \in \R^d$ in the output instance have the form
\[
B = 
\left(
 \begin{array}{c}
B_1 \\
\vdots \\
B_m \\
2 \alpha^{1/p} \cdot I_n
 \end{array}
\right), \qquad
\vec{t} = 
\left(
 \begin{array}{c}
\vec{t}_1 \\
\vdots \\
\vec{t}_m \\
\alpha^{1/p} \cdot \vec{1}_n
 \end{array}
\right),
\]
with blocks $B_i \in \R^{d^* \times n}$ and $\vec{t}_{i} \in \R^{d^*}$ for $1 \leq i \leq m$ and $\alpha := m + (m - W)\delta$. Note that $\alpha$ is the maximum possible contribution of the clauses $C_1, \ldots, C_m$ to $\norm{B\vec{y} - \vec{t}}_p^p$ when $(\Phi, W)$ is a `YES' instance.
For every $1 \leq i \leq m$ and $1 \leq j \leq n$, set the $j$th column $(B_i)_j$ of block $B_i$ (corresponding to the clause $C_i = \lor_{s=1}^k \ell_{i, s}$) as
\[
  (B_i)_j :=\left \{
\begin{array}{rl} 
\vec{v}_s & \textrm{if $x_j$ is the $s$th literal of clause $i$}, \\
- \vec{v}_s & \textrm{if $\lnot x_j$ is the $s$th literal of clause $i$}, \\
\vec{0}_{d^*}  & \textrm{otherwise},
\end{array}
\right.
\]
and set $\vec{t}_i := \vec{t}^* - \sum_{s \in N_i} \vec{v}_s$. Set $r := (\alpha (n + 1))^{1/p}$.

Clearly, the reduction runs in polynomial time.
We next analyze for which $\vec{y} \in \Z^n$ it holds that $\norm{B\vec{y} - \vec{t}}_p \leq r$. Given $\vec{y} \notin \set{0, 1}^n$, 
\[
\norm{B\vec{y} - \vec{t}}_p^p \geq \norm{2\alpha^{1/p} I_n \vec{y} - \alpha^{1/p} \vec{1}_n}_p^p \geq \alpha (n + 2) > r^p
\; ,
\]
so we only need to analyze the case when $\vec{y} \in \set{0, 1}^n$. Consider an assignment $\vec{y} \in \set{0, 1}^n$ to the variables of $\Phi$. Then,
\begin{align*}
\norm{B_i \vec{y} - \vec{t}_i}_p &= \Big\|\sum_{s \in P_i} y_{\ind(\ell_{i, s})} \cdot \vec{v}_s - \sum_{s \in N_i} y_{\ind(\ell_{i, s})} \cdot \vec{v}_s - \Big(\vec{t}^* - \sum_{s \in N_i} \vec{v}_s \Big)\Big\|_p \\
          &= \Big\|\sum_{s \in P_i} y_{\ind(\ell_{i, s})} \cdot \vec{v}_s + \sum_{s \in N_i} \big(1 - y_{\ind(\ell_{i, s})} \big) \cdot \vec{v}_s - \vec{t}^* \Big\|_p \\
          &= \Big\|\sum_{s \in S_i(\vec{a})} \vec{v}_s - \vec{t}^* \Big\|_p.
\end{align*}
By assumption, the last quantity is equal to $1$ if $|S_i(\vec{y})| \geq 1$, and is equal to $(1 + \delta)^{1/p}$ otherwise.
Because $|S_i(\vec{y})| \geq 1$ if and only if $C_i$ is satisfied, it follows that 
\[
\norm{B\vec{y} - \vec{t}}_p^p = \Big(\sum_{i=1}^{m} \norm{B_i \vec{y} - \vec{t}_i}_p^p\Big) + \alpha n 
= m + (m - m^+(\vec{y}))\delta + \alpha n
\; .
\]
Therefore, $\norm{B\vec{y} - \vec{t}}_p \leq r$ if and only if $m^+(\vec{y}) \geq W$, and therefore there exists $\vec{y}$ such that $\norm{B \vec{y} - \vec{t}}_p \leq r$ if and only if $(\Phi, W)$ is a `YES' instance of Max-$k$-SAT, as needed.

To extend this to a reduction from \emph{weighted} Max-$k$-SAT to $\CVP_p$, simply multiply each block $B_i$ and the corresponding target vector $\vec{t}_i$ by $w(C_i)^{1/p}$, where $w(C_i)$ denotes the weight of the clause $C_i$. Then, by adjusting $\alpha$ to depend on the weights $w(C_i)$ we obtain the desire reduction.
\end{proof}

Because the rank $n$ of the output $\CVP_p$ instance matches the number of variables in the input SAT formula, we immediately get the following corollary.

\begin{corollary}
For any efficiently computable $p = p(n) \in [1, \infty)$ if there exists a computable $(p, k)$-isolating parallelepiped for infinitely many $k \in \Z^+$, then, for every constant $\eps > 0$ there is no $2^{(1 - \eps)n}$-time algorithm for $\CVP_p$ assuming W-Max-SAT-SETH. In particular there is no $2^{(1 - \eps)n}$-time algorithm for $\CVP_p$ assuming SETH.
\label{cor:seth-hardness}
\end{corollary}

It is easy to construct a (degenerate) family of isolating parallelepipeds for $p = 1$, and therefore we get hardness of $\CVP_1$ as a simple corollary. (See Figure~\ref{fig:l2-parallelpiped}.)

\begin{corollary}
For every constant $\eps > 0$ there is no $2^{(1 - \eps)n}$-time algorithm for $\CVP_1$ assuming W-Max-SAT-SETH, and in particular there is no $2^{(1 - \eps)n}$-time algorithm for $\CVP_1$ assuming SETH.
\label{cor:cvp1-seth-hardness}
\end{corollary}

\begin{proof}
Let $k \in \Z^+$, let $V = [\vec{v}_1, \ldots, \vec{v}_k]$ with $\vec{v}_1 = \cdots = \vec{v}_k := \frac{1}{k-1} (1, 1)^T \in \R^2$, and let $\vec{t}^* := \frac{1}{k-1} (k, 1)^T \in \R^2$. Then, $\norm{V\vec{x} - \vec{t}^*}_1 = 1$ for every $\vec{x} \in \set{0, 1}^k \setminus \{\vec0_k\}$, and $\norm{\vec{t}^*}_1 = (k + 1)/(k - 1) > 1$. The result follows by Corollary~\ref{cor:seth-hardness}.
\end{proof}

\section{Finding isolating parallelepipeds}
\label{sec:SETH}
We now show how to find a $(p,k)$-isolating parallelepiped given by $V \in \R^{d^* \times k}$ and $\vec{t}^* \in \R^{d^*}$ as in Definition~\ref{def:piped}. We will first show a general strategy for trying to find such an object for any $p \geq 1$ and integer $k \geq 2$. In Section~\ref{sec:odd_p}, we will show how to successfully implement this strategy in the case when $p$ is an odd integer. In Section~\ref{sec:limitations}, we show that $(p,k)$-isolating parallelepipeds do not exist for even integers $p \leq k-1$. Finally, in Section~\ref{sec:all_p} we show how to mostly get around this issue in order to find $(p,k)$-isolating parallelepipeds for ``almost all'' $p \geq 1$.

It will actually be convenient to find a slightly different object that ``works with $\{\pm 1\}^k$ instead of $\{0,1\}^k$.'' We observe below that this suffices.

\begin{lemma}
	\label{lem:pmone}
There is an efficient algorithm that takes as input a matrix $V \in \R^{d^* \times k}$ and vector $\vec{t}^* \in \R^{d^*}$ such that 
	$
	\|V \vec{y} - \vec{t}^*\|_p = 1
	$
	for any $\vec{y} \in \{\pm 1\}^k \setminus \{-\vec{1}_k\}$ and 
	$
	\|-V \vec{1}_k-\vec{t}^*\|_p > 1
	\; ,
	$
and outputs a matrix $V' \in \R^{d^* \times k}$ and vector $(\vec{t}^*)' \in \R^{d^*}$ that form a $(p,k)$-isolating parallelepiped.
\end{lemma}
\begin{proof}
Define $V':=2V$ and $(\vec{t}^*)'=V\vec{1}_k+\vec{t}^*$. Now consider the affine transformation $f\colon \R^k \to \R^k$ defined by $f(\vec{x}):=(2\vec{x}-\vec{1}_k)$, which maps $\{0,1\}^k$ to $\{\pm1\}^k$ and $\vec{0}_k$ to $-\vec{1}_k$.
Then, for $\vec{x}\in\{0,1\}^k$ and $\vec{y}=f(\vec{x})=2\vec{x}-\vec{1}_k\in\{\pm1\}^k$, we have
\begin{equation*}
\|V'\vec{x}-(\vec{t}^*)'\|_p
=\Big\|V'\frac{\vec{y}+\vec{1}_k}{2}-(\vec{t}^*)'\Big\|_p
=\Big\|V'\frac{\vec{y}}{2}+V'\frac{\vec{1}_k}{2}-(\vec{t}^*)' \Big\|_p
=\|V\vec{y}-\vec{t}^* \|_p
\; ,
\end{equation*}
as needed.
\end{proof}

Intuitively, a ``reasonable'' matrix $V$ should act symmetrically on bit strings. I.e., if $\vec{y}, \vec{y}' \in \{\pm 1 \}^k$ have the same number of positive entries, then $V \vec{y}$ should be a permutation of $V \vec{y}'$. This implies that any row of $V$ must be accompanied by all possible permutations of this row. If we further require that each row in $V$ is $\alpha \cdot \vec{v}$ for some $\vec{v} \in \{\pm 1\}^k$ and $\alpha \in \R$, then we arrive at a very general construction that is still possible to analyze.

For weights $\alpha_0, \ldots, \alpha_k \geq 0$, we define $V := V(\alpha_0, \ldots, \alpha_k) \in \R^{2^k \times k}$ as follows. The rows of $V$ are indexed by the strings $\{\pm 1\}^k$, and row $\vec{v}$ is $\alpha_{k-|\vec{v}|}^{1/p} \vec{v}^T$, where 
\[|\vec{v}|
:= \{ \# \text{ of positive entries in \vec{v}}\}
\] is the number of positive entries in $\vec{v}$. For a shift $t^* \in \R$, we set $\vec{t}^* := \vec{t}^*(\alpha_0, \ldots, \alpha_k, t^*)  \in \R^{2^k}$ such that the coordinate of $\vec{t}^*$ corresponding to $\vec{v}$ is $\alpha_{k-|\vec{v}|}^{1/p} t^*$. (Figure~\ref{fig:ell_3_k=3} is an example of this construction. In particular, it shows $V(\alpha_0, \alpha_1, \alpha_2, \alpha_3)$ with $\alpha_0 = 12$, $\alpha_1 = \alpha_2 = 1$ and $\alpha_3 = 0$ and $\vec{t}^*(\alpha_0, \alpha_1,\alpha_2,\alpha_3,t^*)$ with $t^*=2$, where we have omitted the last row, whose weight is zero.)

In what follows, we will use the binomial coefficient $\binom{i}{j}$ extensively, and we adopt the convention that $\binom{i}{j} = 0$ if $j > i$ or $j < 0$ or $j \notin \Z$.

\begin{lemma}
	\label{lem:alpha_norm_thing}
	For any $\vec{y} \in \{\pm 1\}^k$, weights $\alpha_0, \ldots, \alpha_k \geq 0$, and shift $t^* \in \R$,
	\begin{equation*}
	\| V\vec{y} - \vec{t}^*\|_p^p = \sum_{j=0}^k \alpha_{k-j}  \sum_{\ell=0}^{k} \binom{|\vec{y}|}{\ell} \binom{k-|\vec{y}|}{j - \ell}  \cdot \big|2|\vec{y}| + 2j - k-4\ell  - t^*\big|^p
	\; ,
	\end{equation*}
	where $V := V(\alpha_0, \ldots, \alpha_k) \in \R^{2^k \times k}$ and $\vec{t}^* := \vec{t}^*(\alpha_0, \ldots, \alpha_k, t^*) $ as above.
	
	In other words, $\| V\vec{y} - \vec{t}^*\|_p^p$ depends only on $|\vec{y}|$, and if $\vec{w} \in (\R^{\geq 0})^{k+1}$ is the vector such that $w_j = \| V\vec{y}' - \vec{t}^*\|_p^p$ for all $\vec{y}' \in \{\pm 1\}$ with $|\vec{y}'| = j$, then
	\[
	\vec{w} = M_k(p,t^*) \left(
		 \begin{array}{c}
		 \alpha_0 \\
		 \alpha_1\\
		 \vdots \\
		 \alpha_ k
		 \end{array}
	\right)
	\; ,
	\]
	where $M_k(p,t^*) \in \R^{(k+1) \times (k+1)}$ is given by 
	\begin{equation}
	\label{eq:M_def}
	M_k(p,t^*)_{i,j} :=  \sum_{\ell=0}^{k} \binom{i}{\ell} \binom{k-i}{j-\ell}  \cdot \big|2i + 2j -k-4\ell - t^*\big|^p
	\; .
	\end{equation}
\end{lemma}
\begin{proof}
	We have
	\[
	\| V\vec{y} - \vec{t}^*\|_p^p = \sum_{j=0}^k \alpha_j  \sum_{|\vec{v}| = k-j} \big|\inner{\vec{v}, \vec{y}} - t^*\big|^p
	\; .
	\]
	Notice that $\inner{\vec{v}, \vec{y}}$ depends only on how many of the $j$ negative entries of $\vec{v}$ align with the positive entries of $\vec{y}$.  In particular,
	\begin{align*}
	\sum_{|\vec{v}| = k-j} \big|\inner{\vec{v}, \vec{y}} - t^*\big|^p &= \sum_{\ell=0}^{k} \binom{|\vec{y}|}{\ell} \binom{k-|\vec{y}|}{j-\ell}  \cdot \big|-\ell + (|\vec{y}|-\ell) + (j-\ell) - (k-|\vec{y}|-j+\ell)  - t^*\big|^p \\
	&= \sum_{\ell=0}^{k} \binom{|\vec{y}|}{\ell} \binom{k-|\vec{y}|}{j-\ell}  \cdot \big|2|\vec{y}| +2j -k - 4\ell - t^*\big|^p
	\; ,
	\end{align*}
	as needed.
\end{proof}

\begin{lemma}
	\label{lem:stochastic}
	For any $t^* \in \R$, the matrix $M_k(p,t^*)$ defined in Eq.~\eqref{eq:M_def} is stochastic. I.e., 
	$
	M_k(p,t^*) \vec{1}_{k+1}
	= 
	\lambda(t^*) \vec{1}_{k+1}$
	for some $\lambda(t^*) \in \R$. Furthermore, $\lambda(t^*) > 0$.
\end{lemma}
\begin{proof}
	We rearrange the sum corresponding to the $i$th entry of $M_k(p,t^*) \vec{1}_{k+1}$, setting $r := (i+j)/2 -\ell$ to obtain
		\begin{align*}
		 \sum_{j=0}^k \sum_{\ell=0}^{k} \binom{i}{\ell} \binom{k-i}{j-\ell}  \cdot \big|2i + 2j -k - 4\ell - t^*\big|^p 
		 &= \sum_{r} |4r-k-t^*|^p \sum_{j=0}^k \binom{i}{(i+j)/2 - r} \binom{k-i}{r + (j-i)/2}\\
		 &= \sum_{r} |4r-k-t^*|^p \sum_{j=0}^k \binom{i}{r + (i-j)/2} \binom{k-i}{r + (j-i)/2}
		 \; .
		 \end{align*}
		Finally, we recall Vandermonde's identity, which says that 
		\[
		 \sum_{j=0}^k \binom{i}{r + (i-j)/2} \binom{k-i}{r + (j-i)/2} = \binom{k}{r}
		\; .
		\]
		Therefore, the summation does not depend on $i$ (and is clearly positive), as needed.
\end{proof}

Lemma~\ref{lem:stochastic} tells us that for any $t^* \in \R$,
$
M_k(p,t^*) (
\vec{1}_{k+1})/\lambda = \vec{1}_{k+1}
$ for some $\lambda > 0$. We wish to show that, for some $t^* \in \R$, we can find $\alpha_0,\ldots, \alpha_k \geq 0$ such that 
$M_k(p,t^*) (
\alpha_0,
\alpha_1,
\ldots,
\alpha_ k
)^T = \vec{1}_{k+1} + \eps \vec{e}_0$ for some $\eps > 0$, where $\vec{e}_0 := (1,0,\ldots, 0)^T$. In order to do this, it suffices to show that $M_k(p,t^*)$ is invertible. Then, we can take 
\[
(\alpha_0, \alpha_1,\ldots, \alpha_k)^T := \vec{1}_{k+1}/\lambda + \eps M_k(p,t^*)^{-1} \vec{e}_0
\; .
\]
If $\eps := (\lambda \cdot \|M_k(p,t^*)^{-1} \vec{e}_0\|_{\infty})^{-1} > 0$, then the $\alpha_i$ must be non-negative. We make this formal in the next proposition.

\begin{proposition}
	\label{prop:invertible_suffices}
	There is an efficient algorithm that takes as input any $p \geq 1$, an integer $k \geq 2$, and $t^* \in \R$ such that $\det(M_k(p,t^*)) \neq 0$, where $M_k(p,t^*)$ is defined as in Eq.~\eqref{eq:M_def} and outputs $V \in \R^{2^k\times k}$ and $\vec{t}^* \in \R^{2^k}$ that define a $(p,k)$-isolating parallelepiped.
\end{proposition}
\begin{proof}
	By Lemma~\ref{lem:pmone}, it suffices to construct a matrix that works for $\vec{y} \in \{\pm 1\}^k$.
	The algorithm behaves as follows on input $k \geq 2$ and $p \geq 1$ and $t^* \in \R$. By Lemma~\ref{lem:stochastic}, $M_k(p,t^*) \vec{1}_{k+1} = \lambda \vec{1}_{k+1}$ for some $\lambda > 0$. Since we are promised that $\det(M_k(p,t^*)) \neq 0$, we see that $M_k(p,t^*)$ is invertible. The algorithm therefore sets
	\begin{equation}
	\label{eq:choose_alpha}
	(\alpha_0,\ldots, \alpha_k)^T := \vec{1}_{k+1}/\lambda + \eps M_k(p,t^*)^{-1}\vec{e}_0
	\; ,
	\end{equation}
	where $\eps := (\lambda\cdot \|M_k(p,t^*)^{-1}\vec{e}_0\|_{\infty})^{-1} > 0$ is chosen to be small enough such that the $\alpha_i$ are all non-negative. Finally, it outputs $V := V(\alpha_0,\ldots,\alpha_k)$ and $\vec{t}^* := \vec{t}^*(\alpha_0,\ldots, \alpha_k, t^*)$ as defined above.
	
	To prove correctness, we note that $V$ and $\vec{t}^*$ have the desired property. Indeed, it follows from the definition of $M_k(p,t^*)$ in Lemma~\ref{lem:alpha_norm_thing} that 
	$\|V \vec{y} - \vec{t}^*\|_p^p$ is the $j$th coordinate of $\vec{w} := M_k(p,t^*) (\alpha_0, \ldots, \alpha_k)^T$, where $j := |\vec{y}|$. But, by Eq.~\eqref{eq:choose_alpha}, we see that the $j$th coordinate of $\vec{w}$ is $1 + \eps$ if $j = 0$ and is $1$ otherwise, as needed.
\end{proof}

\subsection{Finishing the proof for odd integer \texorpdfstring{$p$}{p}}
\label{sec:odd_p}

We now handle the case when $p \geq 1$ is an odd integer. Notice that, if $p \geq 1$ is an integer, then $\det(M_k(p,t^*))$ is some piecewise combination of polynomials of degree at most $(k+1)p$ in $t^*$. In particular, it is a polynomial in $t^*$ if we restrict our attention to the interval $t^* \in [-k,-k+2]$. We wish to argue that this is not the zero polynomial when $p$ is odd. To prove this, it suffices to show that the coefficient of $(t^*)^{(k+1)p}$ is non-zero, which we do below by studying a matrix whose determinant is this coefficient (when $p$ is odd).

We first show an easy claim concerning matrices that can be written as sums of the identity plus a certain kind of rank-one matrix.

\begin{claim}
	\label{clm:rank_one_plus_ID}
	For any matrix $A \in \R^{d \times d}$ with constant columns given by
	$A_{i,j} = a_j$ for some $a_0,\ldots, a_{d-1} \in \R$ and any $\lambda \in \R$, 
	\[
	\det(A - \lambda I_d) = (-\lambda)^{d-1} \Big( \sum_j a_j - \lambda\Big)
	\; .
	\]
\end{claim}
\begin{proof}
	Notice that $A$ is a rank-one stochastic matrix with one non-zero eigenvalue given by $\sum a_j$. Therefore, the characteristic polynomial of $A$ is $\det(A - \lambda I_d) = (-\lambda)^{d-1} (\sum a_j - \lambda )$, as needed.
\end{proof}

\begin{lemma}
	\label{lem:top_coeff}
	For an integer $k \geq 1$ and an odd integer $p \geq 1$, the function $t^* \mapsto \det(M_k(p,t^*))$, where $M_k(p,t^*)$ is defined as in Eq.~\eqref{eq:M_def}, is a polynomial of degree at most $(k+1)p$ when restricted to the interval $t^* \in [-k,-k+2]$. Furthermore, the coefficient of $(t^*)^{(k+1)p}$ of this polynomial is exactly $2^k(2-2^k)$.
	
	In particular, $t^* \mapsto \det(M_k(p,t^*))$ is a non-zero polynomial of degree $(k+1)p$ on the interval $t^* \in [-k,-k+2]$ for $k \geq 2$.
\end{lemma}
\begin{proof}
	 For any $t^* \in [-k,-k+2]$, the matrix $M_k(p,t^*)$ is given by
	\[
	M_k(p,t^*)_{i,j} =  \sum_{\ell=0}^{k} \binom{i}{\ell} \binom{k-i}{j-\ell}  \cdot \big|2i + 2j -k - 4\ell - t^*\big|^p = \sum_{\ell=0}^{k} \delta_{i+j-2\ell}\binom{i}{\ell} \binom{k-i}{j-\ell}  \cdot (2i + 2j -k -4\ell - t^*)^p
	\; ,
	\]
	where $\delta_{r} = -1$ if $r= 0$ and $1$ otherwise. (Here, we have used the fact that $\binom{i}{\ell} \binom{k-i}{j-\ell} $ is only non-zero when $\ell \leq \min(i,j)$. Therefore, $2i+2j \geq 4\ell$, so that $2i+2j-k-4\ell - t^* \geq 2i+2j-4\ell - 2 \geq 0$ unless $2i-2j-4\ell = 0$.)
	
	The coefficient of $(t^*)^{(k+1)p}$ in the polynomial $t^* \mapsto \det(M_k(p,t^*))$ is therefore given by $\det(M')$, where $M'$ is defined as
	\begin{align*}
	 (M')_{i,j} &:=  -\sum_{\ell=0}^{k} \delta_{2\ell - i - j}\binom{i}{\ell} \binom{k-i}{j-\ell}   \\
	 &= 2 \binom{i}{(i+j)/2} \binom{k-i}{(j-i)/2}- \sum_{\ell=0}^{k} \binom{i}{\ell} \binom{k-i}{j-\ell}   \\
	 &= 2 \binom{i}{(i-j)/2} \binom{k-i}{(j-i)/2}-\binom{k}{j}
	 \; ,
	\end{align*}
	where we have again applied Vandermonde's identity.
	Notice that the first term is non-zero if and only if $i = j$, in which case it is equal to $2$. In other words, $M' =  A+2I_{k+1}$, where $A_{i,j} := -\binom{k}{j}$. The result then follows from Claim~\ref{clm:rank_one_plus_ID}.
\end{proof}

\begin{corollary}
	\label{cor:non-zero_t}
	There is an efficient algorithm that takes as input an integer $k \geq 2$ and odd integer $p \geq 1$ and outputs $t^* \in \Q$ such that $\det(M_k(p,t^*)) \neq 0$, with $M_k(p,t^*)$ defined as in Eq.~\eqref{eq:M_def}.
\end{corollary}
\begin{proof}
	The algorithm works as follows. It chooses $(k+1)p + 1$ distinct points $t_0, \ldots, t_{(k+1)p} \in [-k,-k+2]$ arbitrarily. (E.g., it chooses $t_i = -k + 2i/((k+1)p)$.) For each $t_i$, it computes $\det(M_k(p,t_i))$. It outputs the first $t_i$ such that the determinant is non-zero.
	
	We claim that $\det(M_k(p,t_i)) \neq 0$ for at least one index $i$. Indeed, by Lemma~\ref{lem:top_coeff}, $t^* \mapsto \det(M_k(p,t^*))$ is a non-zero polynomial of degree $(k+1)p$. The result then follows from the fact that such a polynomial can have at most $(k+1)p$ roots.
\end{proof}

Theorem~\ref{thm:main} for finite $p$ now follows immediately from Theorems~\ref{thm:isolating_implies_reduction} together with Proposition~\ref{prop:invertible_suffices}, and Corollary~\ref{cor:non-zero_t}.

\subsection{Limitations of the approach}
\label{sec:limitations}
In the previous section, we showed that for every odd $p\geq1$ and every integer $k \geq 2$, there exists a $(p,k)$-isolating parallelepiped. This allowed us to conclude that $\CVP_p$ is SETH-hard for odd values of $p$. Now, we show that this approach necessarily fails for even $p \geq 2$. Namely, we show that for every even $p$, there is no $(p,k)$-isolating parallelepiped for any $k>p$.%
\footnote{When $k\leq p$, it is possible to construct $(p,k)$-isolating parallelepiped for even $p$. See, e.g., Figure~\ref{fig:l2-parallelpiped}.}
For simplicity, we show this for $p=2$, but a straightforward generalization works for all even $p$.

\begin{lemma}
	\label{lem:incl-excl}
For any integer $k\geq 3$ and vectors $\vec{v}_1,\ldots,\vec{v}_k, \vec{t}^* \in \R^{d^*}$,
we have
\[
\label{eq:incl-excl}
\sum_{S\subseteq[k]}(-1)^{|S|} \, \Big\|\vec{t}^* - \sum_{i\in S}\vec{v}_i\Big\|_2^2=0 \; .
\]
\end{lemma}
\begin{proof}
	We have
\begin{align*}
\sum_{S\subseteq[k]}(-1)^{|S|} \, \Big\| \vec{t}^* - \sum_{i\in S}\vec{v}_i\Big\|_2^2
&=\sum_{S\subseteq[k]}(-1)^{|S|} \, \left( \norm{\vec{t}^*}_2^2-2\Big\langle{\vec{t}^*, \sum_{i\in S}\vec{v}_i}\Big\rangle+\Big\| \sum_{i\in S}\vec{v}_i \Big\|_2^2\right)
\\
&=
\|\vec{t}^*\|_2^2 \cdot \sum_{S\subseteq[k]}(-1)^{|S|}-2\sum_{i\in[k]}\big\langle{\vec{t}^*,\vec{v}_i }\big\rangle\cdot \sum_{S\ni i}(-1)^{|S|}\\
&\qquad +
\sum_{i\in[k]} \|\vec{v_i}\|_2^2 \cdot \sum_{S \ni i}(-1)^{|S|}+2\sum_{i<j}\inner{\vec{v}_i,\vec{v}_j}\cdot \sum_{S \ni i,j}(-1)^{|S|}\\
&=
\|\vec{t}^*\|_2^2 \cdot 0-2\sum_{i\in[k]}\inner{\vec{t}^*,\vec{v}_i }\cdot 0\\
&\qquad +
\sum_{i\in[k]} \|\vec{v_i}\|_2^2 \cdot 0+2\sum_{i<j}\inner{\vec{v}_i,\vec{v}_j}\cdot 0\\
&=0 \; ,
\end{align*}
where the penultimate equality uses the fact that
\[
\sum_{S\subseteq [n]} (-1)^{|S|} = (1-1)^n=0
\]
for $n\geq 1$.
\end{proof}

\begin{corollary}
There is no $(2,k)$-isolating parallelepiped for any integer $k \geq 3$. 
\label{cor:no-2-iso-pipeds}
\end{corollary}
\begin{proof}
Assume $V = [\vec{v}_1, \ldots, \vec{v}_{k}] \in \R^{d^* \times k}$ and $\vec{t}^* \in \R^{d}$ form a $(2,k)$-isolating parallelepiped.
 For any $S\neq\emptyset$, $\norm{\vec{t}^*-\sum_{i\in S}\vec{v}_i}_2^2=1$ by the definition of an isolating parallelepiped. Thus, applying Lemma~\ref{lem:incl-excl}, we have
\[
\norm{\vec{t}^*}_2^2=\sum_{\emptyset\neq S\subseteq[k]}(-1)^{|S|+1} \, \Big\|\vec{t}^*-\sum_{i\in S}\vec{v}_i\Big\|_2^2=1 \; ,
\]
which contradicts the assumption that $V$ and $\vec{t}^*$ form an isolating parallelepiped.
\end{proof}

\subsection{Extending our result to almost all \texorpdfstring{$p$}{p}}
\label{sec:all_p}
We now wish to extend Theorem~\ref{thm:main} to arbitrary $p \geq 1$. Unfortunately, we know that we cannot do this for \emph{all} $p$, since we showed in Section~\ref{sec:limitations} that no such construction is possible when $p$ is an even integer. However, we show a construction that works for ``almost all values of $p$.'' In particular, for any fixed $k$, the construction works for all but finitely many choices of $p$. 
We also observe that this implies that, for every fixed $p_0, k$, there is an $\eps > 0$ such that the construction works for every $p \in (p_0 - \eps)$ or $p \in (p_0 + \eps)$. In particular, for any non-zero $\delta = \delta(n) = o(1)$, the construction works for $p = p_0 + \delta(n)$ for sufficiently large integers $n$.

In Section~\ref{sec:odd_p}, we observed that the function $t^* \mapsto \det(M_k(p,t^*))$ is a piecewise polynomial when $p$ is an integer. This is what allowed us to analyze this case relatively easily (in both Section~\ref{sec:odd_p} and in Section~\ref{sec:limitations}). For non-integer $p$, the function $t^* \mapsto \det(M_k(p,t^*))$ is much less nice. So, instead of holding $p$ fixed and varying $t^*$, we will be interested in studying the function 
$
f_{k,t^*}(p) := \det(M_k(p,t^*))
$
for fixed $t^*$ and $k$. 
We first observe that this function has a fairly nice structure.

\begin{lemma}
	\label{lem:sum_of_exp}
	For any $t^* \in \R$, integer $k \geq 1$, and $p \geq 1$, let 
	\begin{equation}
			\label{eq:f_def}
				f_{k,t^*}(p) := \det(M_k(p,t^*))
				\; ,
	\end{equation}
	where $M_k(p,t^*)$ is as defined in Eq.~\eqref{eq:M_def}. Then, for fixed $k,t^*$, $f_{k,t^*}(p)$ is a Dirichlet polynomial. I.e., there are some real numbers $b_0,\ldots, b_{r}, c_0,\ldots, c_{r} \in \R$ (depending on $t^*$ and $k$) such that
		\begin{equation}
		\label{eq:sum_of_exp}
				f_{k,t^*}(p) = \sum_{i=0}^{r} b_i \exp(c_i p)
		\end{equation}
		for some finite $r$.
\end{lemma}
\begin{proof}
	To see that $f_{k,t^*}(p)$ is a Dirichlet polynomial for fixed $t^*,k$, it suffices to note that (1) each entry of $M_k(p,t^*)$ is a Dirichlet polynomial; (2) the determinant of a matrix can be written as a polynomial in the coordinates; and (3) a polynomial of Dirichlet polynomials is itself a Dirichlet polynomial.
\end{proof}

\begin{corollary}
	\label{cor:non-zero_t_p}
	There is an efficient algorithm that takes as input $k \geq 2$ and any $p \geq 1$ and either fails or outputs $V \in \R^{2^k \times k}$ and $\vec{t}^* \in \R^{2^k}$ that define a $(p,k)$-isolating parallelepiped. Furthermore, for any fixed $k \geq 2$, the algorithm only fails for finitely many choices of $p \geq 1$. 
\end{corollary}
\begin{proof}
	By Corollary~\ref{cor:non-zero_t}, for any $k \geq 2$, we can find a $t^* \in \Q$ such that, say, $f_{k,t^*}(1) \neq 0$, where $f_{k,t^*}(p)$ is defined as in Eq.~\eqref{eq:f_def}. Clearly, $f_{k,t^*}(p)$ is non-zero as a function of $p$ for these values of $t^*, k$. Furthermore, by Lemma~\ref{lem:sum_of_exp}, $f_{k,t^*}(p)$ is a Dirichlet polynomial. The result follows by the fact that any non-zero Dirichlet polynomial has finitely many roots (see, e.g., Theorem 3.1 in~\cite{jameson_2006}).
\end{proof}

\begin{theorem}
	\label{thm:most_p}
	There is an efficient algorithm that takes as input an integer $k \geq 2$ and any $p \geq 1$ and either fails or outputs $V \in \R^{2^k \times k}$ and $\vec{t}^* \in \R^{2^k}$ that define a $(p,k)$-isolating parallelepiped. Furthermore, for any fixed $k \geq 2$, the algorithm only fails on finitely many values of $p \geq 1$.\footnote{As we observed in Section~\ref{sec:limitations}, the set of failure points necessarily includes all even integers $p \leq k-1$.}
\end{theorem}
\begin{proof}
	The result follows immediately from Proposition~\ref{prop:invertible_suffices} and Corollary~\ref{cor:non-zero_t_p}.
\end{proof}

Item~\ref{item:all_but_finite} of Theorem~\ref{thm:all_p_intro} now follows immediately from Theorem~\ref{thm:isolating_implies_reduction} and Theorem~\ref{thm:most_p}.

We now provide what amounts to a different interpretation of the above. %

\begin{lemma}
	\label{lem:non-zero_t_p_plus_eps}
	There is an efficient algorithm that takes as input any $p_0 \geq 1$ and an integer $k \geq 2$ and outputs a value $t^*$ such that $f_{k,t^*}(p_0 + \delta)$ and $f_{k,t^*}(p_0 - \delta)$ are non-zero for sufficiently small $\delta > 0$, where $f_{k,t^*}(p)$ is as defined in Eq.~\eqref{eq:f_def}.
\end{lemma}
\begin{proof}
	The algorithm simply calls the procedure from Corollary~\ref{cor:non-zero_t} with, say, $p = 1$ and outputs the result. I.e., it suffices to choose any $t^*$ such that $f_{k,t^*}(1) \neq 0$. As in the proof of Corollary~\ref{cor:non-zero_t_p}, we observe that the function $f_{k,t^*}(p)$ is zero on a finite set of values $X$. The result then follows by taking $\delta := \min_{x \in X \setminus \set{p_0}} |x - p_0|/2$ if $X \setminus \set{p_0}$ is non-empty, and $\delta := c$ for any $c > 0$ otherwise.

\end{proof}

Finally, we derive the main theorem of this section.

\begin{theorem}
	\label{thm:piped_p_plus_eps}
	For any efficiently computable $\delta(n) \neq 0$ that converges to zero as $n \to \infty$ and $p_0 \geq 1$, there is an efficient algorithm that takes as input an integer $k \geq 2$ and sufficiently large positive integer $n$ and outputs a matrix $V \in \R^{2^k \times k}$ and vector $\vec{t}^* \in \R^{2^k}$ that define a $(p_0 + \delta(n),k)$-isolating parallelepiped.
\end{theorem}
\begin{proof}
	The result follows immediately from Proposition~\ref{prop:invertible_suffices} and Lemma~\ref{lem:non-zero_t_p_plus_eps}. In particular, the algorithm runs the procedure from Lemma~\ref{lem:non-zero_t_p_plus_eps}, receiving as output some $t^* \in \R$ such that $f_{k,t^*}(p_0 \pm \eps)$ is non-zero for sufficiently small $\eps > 0$. In particular, if $n$ is sufficiently large, then $f_{k,t^*}(p_0 + \delta(n))$ will be non-zero. The result then follows from Proposition~\ref{prop:invertible_suffices}.
\end{proof}

Item~\ref{item:p_plus_eps} of Theorem~\ref{thm:all_p_intro} now follows from Theorem~\ref{thm:isolating_implies_reduction} and Theorem~\ref{thm:piped_p_plus_eps}.

\section{Gap-ETH-based Hardness of Approximation}
\label{sec:approx_hard}
In this section we prove Gap-ETH-based hardness of approximation for $\CVP_p$ for every $p \in [1, \infty)$. We also show a stronger hardness of approximation result for $\CVP_1$. (In Section~\ref{sec:infinity!}, we additionally show a stronger result for $p = \infty$.)
The main idea behind our proof is to use the reduction from Max-$2$-SAT to $\CVP$ described in Section~\ref{subsec:techniques} without the ``identity matrix gadget'' that we used to force any closest vector to be a 0-1 combination of basis vectors. We will show that this is permissible because the resulting CVP instance always has a 0-1 combination of basis vectors that is at least as close to the target as any other lattice vector.

\begin{theorem}
For every $0 \leq \delta < \eps \leq 1$ and every $p \in [1, \infty)$, there is a polynomial-time reduction from any $(\delta, \eps)$-Gap-$2$-SAT instance with $n$ variables to an instance of $\gamma$-$\CVP_p$ of rank $n$, where
\[
\gamma := \left(\frac{\delta + (1 - \delta) \cdot 3^p}{\eps + (1 - \eps) \cdot 3^p} \right)^{1/p} > 1
\; .
\]
In particular, when $\eps = 1$, the corresponding approximation factor is 
$
(\delta + (1-\delta) \cdot 3^p)^{1/p}
$.
\label{thm:gap2sat-to-gamma-cvpp}
\end{theorem}

\begin{proof}
Given a $(\delta, \eps)$-Gap-$2$-SAT instance with $n$ variables $x_1, \ldots, x_n$ and $m$ clauses $C_1, \ldots, C_m$, we construct a $\CVP_p$ instance $(B, \vec{t}, r)$ for some fixed $p \in [1, \infty)$ as follows.
Let $B \in \Z^{m \times n}$ be the basis defined by
	\[
	b_{i,j} := \left\{
	\begin{array}{rl}
	2 & \textrm{if $C_i$ contains $x_j$,}\\
	-2 & \textrm{if $C_i$ contains $\lnot x_j$,}\\
	0 & \textrm{otherwise}.
	\end{array}\right.
	\;
	\]
Let $\vec{t} \in \R^m$ be the target vector defined by $t_i := 3 - 2 |N_i|$,
and let $r := (\eps m + (1 - \eps) m \cdot 3^p)^{1/p}$. 

We claim that there is always a 0-1 combination of basis vectors which is a closest vector to $\vec{t}$. Assuming this claim, we analyze $\norm{B \vec{y} - \vec{t}}_p$ only for $\vec{y} \in \set{0, 1}^n$ without loss of generality, while deferring the proof of the claim until the end.

Let $\vec{y} \in \set{0, 1}^n$ be an assignment to the variables of $\Phi$. For every $1 \leq i \leq m$,
\begin{align}
\begin{split}
|(B\vec{y} - \vec{t})_i| &= \Big| 2 \Big(\sum_{s \in P_i} y_{\ind(\ell_{i,s})} - \sum_{s \in N_i} y_{\ind(\ell_{i,s})} \Big) - (3 - 2 |N_i|) \Big| \\
&= \Big| 2 \Big( \sum_{s \in P_i} y_{\ind(\ell_{i,s})} + \sum_{s \in N_i} (1 - y_{\ind(\ell_{i,s})}) \Big) - 3 \Big| \\
&= \big| 2 \cdot |S_i(\vec{y})| - 3\big| 
\; .
\end{split}
\label{eq:approx-reduction-coordinate}
\end{align}
The last expression is equal to $1$ if $\vec{y}$ satisfies $C_i$ (i.e. if $|S_i(\vec{y})| \geq 1$) and $3$ if not. %
We therefore have
\[
\norm{B\vec{y} - \vec{t}}_p^p = \sum_{i=1}^m |(B\vec{y} - \vec{t})_i|^p = m^+(\vec{y}) + (m - m^+(\vec{y})) \cdot 3^p .
\]
Therefore, if $\val(\Phi) \geq \eps$ then there exists $\vec{y} \in \set{0, 1}^n$ such that $\norm{B\vec{y} - \vec{t}}_p^p \leq \eps m + (1 - \eps) m \cdot 3^p = r^p$, and if $\val(\Phi) < \delta$, then for every $\vec{y} \in \set{0, 1}^n$, it holds that $\norm{B\vec{y} - \vec{t}}_p^p > \delta m + (1 - \delta) m \cdot 3^p = \gamma^p r^p$. It follows that the reduction achieves the claimed approximation factor of $\gamma$.

It remains to prove the claim that there is always a 0-1 combination of basis vectors which is a closest vector to $\vec{t}$. 
We show this by demonstrating that for every $\vec{y} \in \Z^n$ there exists $\chi(\vec{y}) \in \set{0, 1}^n$ such that $\norm{B \cdot \chi(\vec{y}) - \vec{t}}_p \leq \norm{B \vec{y} - \vec{t}}_p$.
Given $\vec{y} \in \Z^n$, let $\chi(\vec{y}) \in \set{0, 1}^n$ denote the vector whose $i$th coordinate is set to 1 if $y_i \geq 1$ and is set to 0 otherwise.

Fix $\vec{y} \in \Z^n$ and $1 \leq i \leq n$, and refer to Equation~\eqref{eq:approx-reduction-coordinate}. If $\chi(\vec{y})$ satisfies $C_i$ then
$|(B\vec{y} - \vec{t})_i| \geq |(B \cdot \chi(\vec{y}) - \vec{t})_i| = 1$.
On the other hand, if $\chi(\vec{y})$ does not satisify $C_i$, then $y_{\ind(\ell_{i,s})} \leq 0$ for all $s \in P_i$ and $1 - y_{\ind(\ell_{i,s})} \leq 0$ for all $s \in N_i$ so that $|(B\vec{y} - \vec{t})_i| \geq 3$ and therefore
$|(B\vec{y} - \vec{t})_i| \geq |(B \cdot \chi(\vec{y}) - \vec{t})_i| = 3$.
Combining these cases it follows that for all $\vec{y} \in \Z^n$ and $1 \leq i \leq n$,
$|(B \cdot \chi(\vec{y}) - \vec{t})_i| \leq |(B\vec{y} - \vec{t})_i|$,
and therefore $\norm{B \cdot \chi(\vec{y}) - \vec{t}}_p \leq \norm{B\vec{y} - \vec{t}}_p$, proving the claim.
\end{proof}

\begin{corollary}
For every $p \in [1, \infty)$ and $0 < \delta < \delta' < 1$,
there is no $2^{o(n)}$-time algorithm for $\gamma$-$\CVP_p$ with %
\[
\gamma := \left( \frac{6 + \delta' + (4 - \delta') \cdot 3^p}{7 + 3^{p+1}} \right)^{1/p} > 1
\]
unless randomized Gap-ETH (with respect to $(\delta, 1)$-Gap-$3$-SAT) fails.
\label{cor:gap-eth-hardness-of-approx}
\end{corollary}

\begin{proof}
Concatenate the reductions described in Proposition~\ref{prop:gap-sat-sparsification}, Proposition~\ref{prop:gap3sat-to-gap2sat}, and Theorem~\ref{thm:gap2sat-to-gamma-cvpp}.\footnote{%
Note that the Gap-2-SAT instance output by Proposition~\ref{prop:gap3sat-to-gap2sat} contains $1$-clauses as well as $2$-clauses. We therefore stress that the reduction described in Theorem~\ref{thm:gap2sat-to-gamma-cvpp} works for this case as well.}
\end{proof}

\subsection{A stronger result for \texorpdfstring{$\ell_1$}{ell 1}}

In Theorem~\ref{thm:gap2sat-to-gamma-cvpp} we showed that there was no need to use the identity matrix gadget in our reduction from Gap-$2$-SAT to $\CVP$. An interesting question is whether the identity matrix gadget is also unnecessary in our reduction in Theorem~\ref{thm:max-k-sat-to-low-rank} from Gap-$k$-SAT to $\CVP_p$ using $(p, k)$-isolating parallelepipeds. If so, then we get stronger ``Gap-SETH-hardness'' for $\CVP_p$ for all $p$ for which there exist $(p, k)$-isolating parallelepipeds for infinitely many $k \in \Z^+$.

Although we do not know how to show this in general, 
we are able to get stronger hardness of approximation for $\CVP_1$ in this way by using the family of $(1, k)$-isolating parallelepipeds described in Corollary~\ref{cor:cvp1-seth-hardness}.
The analysis is similar to the analysis in Theorem~\ref{thm:gap2sat-to-gamma-cvpp}. In particular, a 0-1 combination of basis vectors will always be at least as close to the target vector as any other lattice vector will.

\begin{theorem}
For every $0 \leq \delta < \eps \leq 1$, there is a polynomial-time reduction from any $(\delta, \eps)$-Gap-$k$-SAT instance with $n$ variables to an instance of $\gamma$-$\CVP_1$ of rank $n$, where
\[
\gamma := \Big(\frac{\delta \cdot (k - 1) + (1 - \delta) \cdot (k + 1)}{\eps \cdot (k - 1) + (1 - \eps) \cdot (k + 1)}\Big) > 1.
\]
\label{thm:cvp1-hardness-of-approx}
\end{theorem}

\begin{proof}
Given a $(\delta, \eps)$-Gap-$k$-SAT instance with $n$ variables $x_1, \ldots, x_n$ and $m$ clauses $C_1, \ldots, C_m$, we construct a $\CVP_1$ instance $(B, \vec{t}, r)$ as follows.
The basis $B \in \Z^{(2m) \times n}$ and target vector $\vec{t} \in \Z^{2m}$ in the output instance have the form
\[
B = 
\left(
 \begin{array}{c}
B_1 \\
\vdots \\
B_m \\
 \end{array}
\right), \qquad
\vec{t} = 
\left(
 \begin{array}{c}
\vec{t}_1 \\
\vdots \\
\vec{t}_m
 \end{array}
\right),
\]
with blocks $B_i \in \Z^{2 \times n}$ and $\vec{t}_i \in \Z^2$ for $1 \leq i \leq m$. For every $1 \leq i \leq m$ and $1 \leq j \leq n$, set the $j$th column $(B_i)_j$ of block $B_i$ (corresponding to the clause $C_i = \lor_{s=1}^k \ell_{i, s}$) as

\[
  (B_i)_j :=\left \{
\begin{array}{rl} 
(1, 1)^T   & \textrm{if $C_i$ contains $x_j$}, \\
- (1, 1)^T & \textrm{if $C_i$ contains $\lnot x_j$}, \\
(0, 0)^T   & \textrm{otherwise},
\end{array}
\right.
\]
and set $\vec{t}_i := (k, 1)^T - |N_i| \cdot (1, 1)^T$. Set $r := \eps m (k - 1) + (1 - \eps) \cdot m (k + 1)$.

We claim that there is always a 0-1 combination of basis vectors that is closest to $\vec{t}$. Assuming the claim, we analyze $\norm{B\vec{y} - \vec{t}}$ only for $\vec{y} \in \set{0, 1}^n$ without loss of generality, while deferring the proof of the claim until the end.

Let $\vec{y} \in \set{0, 1}^n$ be an assignment to the variables of $\Phi$. For every $1 \leq i \leq m$,

\begin{align}
\begin{split}
\norm{B_i \vec{y} - \vec{t}_i}_1 &= \Big\|\sum_{s \in P_i} y_{\ind(\ell_{i, s})} \cdot (1, 1)^T - \sum_{s \in N_i} y_{\ind(\ell_{i, s})} \cdot (1, 1)^T - \Big((k, 1)^T - |N_i| \cdot (1, 1)^T \Big)\Big\|_1 \\
          &= \Big\|\sum_{s \in P_i} y_{\ind(\ell_{i, s})} \cdot (1, 1)^T + \sum_{s \in N_i} \big(1 - y_{\ind(\ell_{i, s})} \big) \cdot (1, 1)^T - (k, 1)^T \Big\|_1 \\
          &= \Big\||S_i(\vec{y})| \cdot (1, 1)^T - (k, 1)^T \Big\|_1.
\label{eq:approx-reduction-coordinate-l1}
\end{split}
\end{align}
The last expression is equal to $k - 1$ if $\vec{y}$ satisfies $C_i$ (i.e. if $|S_i(\vec{y})| \geq 1$) and $k + 1$ if not. We therefore have 

\[
\norm{B\vec{y} - \vec{t}}_1 = \sum_{i=1}^m \| B_i\vec{y} - \vec{t}_i \|_1 = m^+(\vec{y}) \cdot (k - 1) + (m - m^+(\vec{y})) \cdot (k + 1).
\]
Therefore, if $\val(\Phi) \geq \eps$ then there exists $\vec{y} \in \set{0, 1}^n$ such that $\norm{B\vec{y} - \vec{t}}_1 \leq \eps m (k - 1) + (1 - \eps) \cdot m (k + 1) = r$,
and if $\val(\Phi) < \delta$ then for every $\vec{y} \in \set{0, 1}^n$, it holds that $\norm{B\vec{y} - \vec{t}}_1 > \delta m (k - 1) + (1 - \delta) m (k + 1) = \gamma r$.
It follows that the reduction achieves the claimed approximation factor of $\gamma$.

It remains to prove the claim that there is always a 0-1 combination of basis vectors that is a closest vector to $\vec{t}$. We show this by demonstrating that for every $\vec{y} \in \Z^n$ there exists $\chi(\vec{y}) \in \set{0, 1}^n$ such that $\norm{B \chi(\vec{y}) - \vec{t}}_1 \leq \norm{B \vec{y} - \vec{t}}_1$. Given $\vec{y} \in \Z^n$, let $\chi(\vec{y}) \in \set{0, 1}^n$ denote the vector whose coordinate is set to $1$ if $y_i \geq 1$ and is set to $0$ otherwise.

Note that $\norm{c (1, 1)^T - (k, 1)^T}_1 = k - 1$ for all $c \in [k]$, and $\norm{c (1, 1)^T - (k, 1)^T}_1 \geq k + 1$ for all $c \in \Z \setminus [k]$. Fix $\vec{y} \in \Z^n$ and $1 \leq i \leq n$, and refer to Equation~\eqref{eq:approx-reduction-coordinate-l1}. If $\chi(\vec{y})$ satisfies $C_i$ then $\norm{B_i \cdot \vec{y} - \vec{t}_i}_1 \geq \norm{B_i \cdot \chi(\vec{y}) - \vec{t}_i}_1 = k - 1$. On the other hand, if $\chi(\vec{y})$ does not satisfy $C_i$, then $y_{\ind(\ell_{i, s})} \leq 0$ for all $s \in P_i$ and $1 - y_{\ind(\ell_{i, s})} \leq 0$ for all $s \in N_i$ so that $\norm{B_i \cdot \vec{y} - \vec{t}_i}_1 \geq \norm{B_i \cdot \chi(\vec{y}) - \vec{t}_i}_1 = k + 1$. Combining these cases it follows that for all $\vec{y} \in \Z^n$ and $1 \leq i \leq n$, $\norm{B_i \cdot \chi(\vec{y}) - \vec{t}_i}_1 \leq \norm{B_i \cdot \vec{y} - \vec{t}_i}_1$, and therefore $\norm{B \cdot \chi(\vec{y}) - \vec{t}}_1 \leq \norm{B\vec{y} - \vec{t}}_1$, proving the claim.

\end{proof}

\section{Additional hardness results}
\label{sec:other}
In this section we prove a number of additional results about the quantitative hardness of $\CVP$ and related problems. In Section~\ref{sec:cvpp}, we give a reduction from Max-$2$-SAT to $\CVPP_p$ for all $p \in [1, \infty)$, proving Theorem~\ref{thm:cvpp_intro}. In Section~\ref{sec:ETH-MAX-CUT}, we give a reduction from Max-$k$-SAT (and in particular Max-$2$-SAT) to $\CVP_p$ for all $p \in [1, \infty)$, proving Theorem~\ref{thm:maxcut}. 
Finally, in Section~\ref{sec:infinity!} we give a reduction from $k$-SAT to $\CVP_{\infty}$ and $\SVP_{\infty}$, proving the special case of Theorem~\ref{thm:main} for $p = \infty$.

Our reductions all use the same high-level idea as the reduction given in Theorem~\ref{thm:max-k-sat-to-low-rank}, but each uses new ideas as well. Throughout this section we adopt the notation from Section~\ref{sec:geom-char}.

\subsection{Hardness of \texorpdfstring{$\CVPP_p$}{CVPP-p}}
\label{sec:cvpp}

In this section, we prove ETH-hardness of $\CVPP$. To do this, for every $n$, we construct a single lattice $\lat_n \subset \R^d$ of rank $O(n^2)$, such that for every $n$-variable instance of Max-$2$-SAT, there exists an efficiently computable $\vec{t} \in \R^d$ that is close to the lattice if and only if $\Phi$ is satisfiable. 
Clearly, any efficient algorithm for $\CVPP$ on this lattice would imply a similarly efficient algorithm for Max-$2$-SAT (and also $3$-SAT, as described below).

Our basis $B_n$ for $\lat_n$ will encode all possible $O(n^2)$ clauses of a Max-$2$-SAT instance on $n$ variables, together with a gadget that will allow us to ``switch on or off'' each clause by only changing the coordinates of the \emph{target vector} $\vec{t}$. (This gadget costs us a quadratic blow-up in the lattice rank.) Then, given an instance $(\Phi,W)$ of Max-$2$-SAT, we define the target vector $\vec{t}$ such that it ``switches on'' all clauses from $\Phi$ and ``switches off'' all the remaining clauses. 

\begin{lemma}
\label{lem:cvpp}
For every $p \in [1, \infty)$, there is a pair of polynomial-time algorithms $(P, Q)$ (in analogy to the definition of $\CVPP$) that behave as follows.
\begin{enumerate}
	\item On input an integer $n \geq 1$, $P$ outputs a basis $B_n \in\R^{d\times N}$ of a rank $N$ lattice $\lat_n \subset \R^d$, where $d = d(n) = O(n^2)$ and $N = N(n) = O(n^2)$.
	\item On input a Max-$2$-SAT instance with $n$ variables, $Q$ outputs a target vector $\vec{t} \in \R^d$ and a distance bound $r \geq 0$ such that $\dist_p(\vec{t}, \lat_n) \leq r$ if and only if the input is a `YES' instance.
\end{enumerate}
\end{lemma}

\begin{proof}

Let $M=4\binom{n}{2}=O(n^2)$ be the total possible number of $2$-clauses on $n$ variables, and let $C_1,\ldots,C_M$ denote those clauses.

The algorithm $P$ constructs the basis $B_n\in\R^{d\times N}$, where $d:=n+2M, N:=n+M$,  as
\[
B := 
\begin{pmatrix}
(\vec{b}_1^T, \vec{c}_1^T) \\
\vdots \\
(\vec{b}_M^T, \vec{c}_M^T) \\
2 \alpha^{1/p} I_{N}
\end{pmatrix}, 
\]
with rows $(\vec{b}_i^T, \vec{c}_i^T)$ of $B$ satisfying $\vec{b}_i \in \R^n$ and $\vec{c}_i \in \R^{M}$ for $1 \leq i \leq M$, and where $\alpha :=2^p M$.
For every $1 \leq i \leq M$ and $1 \leq j \leq n$, set the $j$th coordinate of $\vec{b}_i$ (corresponding to the clause $C_i = \ell_{i,1} \lor \ell_{i,2}$) as

\[
(\vec{b}_i)_j :=\left \{
\begin{array}{rl} 
1 & \text{if $x_j$ appears in the $i$th clause}, \\
-1 & \text{if $\lnot x_j$ appears in the $i$th clause}, \\
0  & \text{otherwise}.
\end{array}
\right.
\]
For every $1 \leq i \leq M$, set $\vec{c}_i^T := (\vec{0}_{(i - 1)}^T, 1,\vec{0}_{(M - i)}^T)$, i.e., set $(\vec{c}_1,\ldots, \vec{c}_M)=I_M$.

Given an instance $(\Phi, W)$ of Max-$2$-SAT with $m$ clauses, the algorithm $Q$ outputs 
\[
r:=\big((M-m+W) \cdot 1/2^p + (m - W)\cdot (3/2)^p+\alpha(n+M-m)\big)^{1/p}
\]
and $\vec{t}\in\R^d$ defined as 
\[
\vec{t} := 
\left(
 \begin{array}{c}
u_1 \\
\vdots \\
u_M \\
\alpha^{1/p}\cdot\vec{1}_{n}\\
 v_1\\
\vdots\\
v_M\\
 \end{array}
\right),
\]
where for $1\leq i\leq M$, $u_i=3/2-|N_{i}|$, and $v_i=0$ if the clause $C_i$ appears in the formula $\Phi$ and $v_i=\alpha^{1/p}$ otherwise.

Clearly both algorithms are efficient. We now analyze for which $\vec{y}\in\Z^N$ we have $\|B\vec{y}-\vec{t}\|_p\leq r$. 
Note that the vector $\vec{v}=(v_1,\ldots,v_M)^T$ has exactly $m$ coordinates equal to zero, and $M-m$ coordinates equal to $\alpha^{1/p}$.
Given $\vec{y} \notin \set{0, 1}^{N}$, we have 
\[
\norm{B\vec{y} - \vec{t}}_p^p \geq \norm{2\alpha^{1/p} I_N \vec{y} - (\alpha^{1/p}\vec{1}_{n}^T,\vec{v}^T)^T}_p^p \geq \alpha(n+M-m+1)\geq 2^p M + \alpha(M+n-m) > r^p
\; .
\]
 Furthermore, if $\vec{y} \in \set{0, 1}^{N}$ has a non-zero coordinate $y_{n+M+i}$ (for $1\leq i\leq  M$) at a position corresponding to a $t_{n+M+i}=0$ in $\vec{t}$ (i.e., $C_{i}\in\Phi$), then again  $\norm{B\vec{y} - \vec{t}}_p^p \geq \alpha(n+M-m+1) > r^p$. So, we can restrict our attention to $\vec{y} \in\set{0, 1}^{N}$ with $y_{n+M+i}=0$ whenever $C_i\in\Phi$.

Consider an assignment $\vec{a} \in \set{0, 1}^n$ to the $n$ variables of $\Phi$. Take $(\vec{y}')^T = (y'_1,\ldots y'_{M})^T \in \set{0, 1}^{M}$, and set $\vec{y} := (\vec{a}^T, (\vec{y}')^T)^T$. Then, for $1 \leq i \leq M$,
\begin{align*}
\Big|\langle(\vec{b}_i^T, \vec{c}_i^T)^T, \vec{y}\rangle - t_i \Big| &= \Big|\sum_{s \in P_i} y_{\ind(\ell_{i, s})} - \sum_{s \in N_i} y_{\ind(\ell_{i, s})} + \langle\vec{c}_i, \vec{y}'\rangle - (3/2 - |N_i|) \Big| \\
                              &= \Big|\sum_{s \in P_i} y_{\ind(\ell_{i, s})} - \sum_{s \in N_i} (1 - y_{\ind(\ell_{i, s})}) + y'_i - 3/2 \Big| \\
                              &= \Big||S_i(\vec{a})| + y'_i - 3/2 \Big| .
\end{align*}
If $C_i\notin\Phi$, then there exists $y'_i \in \set{0,1}$, such that $|\langle(\vec{b}_i^T, \vec{c}_i^T)^T, \vec{y}\rangle - t_i |=1/2$. Moreover, the choice of $y'_i$ does not affect $|\langle(\vec{b}_i^T, \vec{c}_i^T)^T, \vec{y}\rangle - t_{i'}|$ for $i' \neq i$.
If $C_i\in\Phi$ and $|S_i(\vec{a})| >0$ for $y'_i=0$, then $|\langle(\vec{b}_i^T, \vec{c}_i^T)^T, \vec{y}\rangle - t_i |=1/2$.
On the other hand, if $C_i\in\Phi$ and $|S_i(\vec{a})| =0$, then $y'_i=0$ implies
$|\langle(\vec{b}_i^T, \vec{c}_i^T)^T, \vec{y}\rangle - t_i| \geq 3/2$. 

Because $|S_i(\vec{a})| \geq 1$ if and only if $C_i$ is satisfied, we see that the following holds if and only if the number of satisfied clauses $m^+(\vec{a})$ is at least $W$: 

There exists a $\vec{y}'$ such that, setting $\vec{y} := (\vec{a}, \vec{y}')$, we have
\begin{align*}
\norm{B\vec{y} - \vec{t}}_p^p &= \Big(\sum_{i=1}^M |\langle(\vec{b}_i^T, \vec{c}_i^T)^T, \vec{y}\rangle - t_i|^p\Big) + \alpha (n+M-m)\\ 
&= (M-m+m^+(\vec{a})) \cdot (1/2)^p + (m - m^+(\vec{a}) ) \cdot (3/2)^p + \alpha (n+M-m)\\
&\leq (M-m+W) \cdot (1/2)^p + (m - W) \cdot (3/2)^p + \alpha (n+M-m)\\
&= r^p
\; .
\end{align*}

Therefore, there exists $\vec{y}$ with $\norm{B\vec{y} - \vec{t}}_p \leq r$ if and only if $(\Phi, W)$ is a `YES' instance of Max-$2$-SAT.
\end{proof}

\begin{proof}[Proof of Theorem~\ref{thm:cvpp_intro}]
The main statement follows from Lemma~\ref{lem:cvpp}. The ``in particular'' part follows from Lemma~\ref{lem:cvpp} and the existence of a reduction from $3$-SAT with $n$ variables clauses to (many) Max-$2$-SAT instances with $O(n)$ variables. Indeed, such a reduction follows by applying the Sparsification Lemma (Proposition~\ref{prop:sparsification}), and then reducing each sparse $3$-SAT instance to a Max-$2$-SAT instance with only a linear blow-up in the number of variables (see, e.g., the reduction in~\cite[Theorem 1.1]{GJS76}).
\end{proof}

\subsection{ETH- and Max-$2$-SAT-hardness for all \texorpdfstring{$p \in [1, \infty)$}{1 <= p < infinity}}
\label{sec:ETH-MAX-CUT}

\begin{theorem}
For every $p = p(n) \in [1, \infty)$ and $k \geq 2$ there is a polynomial-time reduction from any (Weighted-)Max-$k$-SAT instance with $n$ variables and $m$ clauses to a $\CVP_p$ instance of rank $n + (k - 2)m$.
\label{thm:max-k-sat-to-high-rank}
\end{theorem}

\begin{proof}
For simplicity we give a reduction from unweighted Max-$k$-SAT. The modification sketched for reducing from Weighted Max-$k$-SAT in Theorem~\ref{thm:max-k-sat-to-low-rank} works here as well.
Namely, we give a reduction from a Max-$k$-SAT instance $(\Phi, W)$ to an instance $(B, \vec{t}, r)$ of $\CVP_p$. Here the formula $\Phi$ is on $n$ variables $x_1, \ldots, x_n$ and $m$ clauses $C_1, \ldots, C_m$. $(\Phi, W)$ is a `YES' instance if there exists an assignment $\vec{a}$ such that $m^+ \geq W$.
We assume without loss of generality that each variable appears at most once per clause. We define the output $\CVP_p$ instance as follows. Let $d := n + (k - 1)m$ and let $N := n + (k - 2)m$. The basis $B \in \R^{d \times N}$ in the output instance has the form

\[
B = 
\begin{pmatrix}
(\vec{b}_1^T, \vec{c}_1^T) \\
\vdots \\
(\vec{b}_m^T, \vec{c}_m^T) \\
2 \alpha^{1/p} \cdot I_{N}
\end{pmatrix}, \qquad
\vec{t} = 
\begin{pmatrix}
t_1 \\
\vdots \\
t_m \\
\alpha^{1/p} \cdot \vec{1}_{N}
\end{pmatrix}
\]
with rows $(\vec{b}_i^T, \vec{c}_i^T)$ of $B$ satisfying $\vec{b}_i \in \R^n$ and $\vec{c}_i \in \R^{m(k-2)}$ for $1 \leq i \leq m$, and where $\alpha := W \cdot (\half)^p + (m - W) \cdot (\frac{3}{2})^p$. For every $1 \leq i \leq m$ and $1 \leq j \leq n$, set the $j$th coordinate of $\vec{b}_i$ (corresponding to the clause $C_i = \lor_{s=1}^k \ell_{i,s}$) as

\[
(\vec{b}_i)_j :=\left \{
\begin{array}{rl} 
1 & \textrm{if $x_j$ is in the $i$th clause}, \\
-1 & \textrm{if $\lnot x_j$ is in the $i$th clause}, \\
0  & \textrm{otherwise}.
\end{array}
\right.
\]
For every $1 \leq i \leq m$, set $\vec{c}_i^T := (\vec{0}_{(i - 1)(k - 2)}^T, - \vec{1}_{k-2}^T, \vec{0}_{(m - i)(k - 2)}^T)$. I.e., each $\vec{c}_i$ has a block of $-1$s of length $k - 2$, and $\vec{c}_i$, $\vec{c}_{i'}$ are coordinate-wise disjoint for $i \neq i'$.
For every $1 \leq i \leq m$ set $t_i := 3/2 - |N_i|$. Finally, set $r := (\alpha (N + 1))^{1/p}$.

We next analyze for which $\vec{y} \in \Z^{N}$ it holds that $\norm{B\vec{y} - \vec{t}}_p \leq r$. Given $\vec{y} \notin \set{0, 1}^{N}$, $\norm{B\vec{y} - \vec{t}}_p^p \geq \norm{2 \alpha^{1/p} I_N \vec{y} - \alpha^{1/p}\vec{1}_{N}}_p^p \geq \alpha (N + 2) > r^p$, so we only need to analyze the case when $\vec{y} \in \set{0, 1}^n$.

Consider an assignment $\vec{a} \in \set{0, 1}^n$ to the variables of $\Phi$, take $(\vec{y}')^T = ((\vec{y}_1')^T, \ldots, (\vec{y}_m')^T) \in \set{0, 1}^{(k-2)m}$ with $\vec{y}_i' \in \set{0, 1}^{k-2}$ for $1 \leq i \leq m$, and set $\vec{y}^T := (\vec{a}^T, (\vec{y}')^T)$. Then, for $1 \leq i \leq m$,
\begin{align*}
\Big|\langle(\vec{b}_i^T, \vec{c}_i^T)^T, \vec{y}\rangle - t_i \Big| &= \Big|\sum_{s \in P_i} y_{\ind(\ell_{i, s})} - \sum_{s \in N_i} y_{\ind(\ell_{i, s})} + \langle\vec{c}_i, \vec{y}_i'\rangle - (3/2 - |N_i|) \Big| \\
                              &= \Big|\sum_{s \in P_i} y_{\ind(\ell_{i, s})} + \sum_{s \in N_i} (1 - y_{\ind(\ell_{i, s})}) - \norm{\vec{y}_i'}_1 - 3/2 \Big| \\
                              &= \Big||S_i(\vec{a})| - \norm{\vec{y}_i'}_1 - 3/2 \Big| .
\end{align*}

It follows that if $|S_i(\vec{a})| = 0$ then $|\langle(\vec{b}_i^T, \vec{c}_i^T)^T, \vec{y}\rangle - t_i| \geq \frac{3}{2}$. On the other hand, if $|S_i(\vec{a})| \geq 1$, then there exists $\vec{y}_i' \in \set{0, 1}^{k-2}$ such that $|\langle(\vec{b}_i^T, \vec{c}_i^T)^T, \vec{y}\rangle - t_i| = \half$. Indeed, picking any $\vec{y}_i'$ with Hamming weight $|S_i(\vec{a})| - 2$ or $|S_i(\vec{a})| - 1$ achieves this. Moreover, the choice of $\vec{y}_i'$ does not affect $|\langle(\vec{b}_i^T, \vec{c}_i^T)^T, \vec{y}\rangle - t_{i'}|$ for $i' \neq i$.

Because $|S_i(\vec{a})| \geq 1$ if and only if $C_i$ is satisfied, we see that the following holds if and only if the number of satisfied clauses $m^+(\vec{a})$ is at least $W$: 

There exists a $\vec{y}'$ such that, setting $\vec{y} := (\vec{a}, \vec{y}')$, we have
\begin{align*}
\norm{B\vec{y} - \vec{t}}_p^p &= \Big(\sum_{i=1}^m |\langle(\vec{b}_i^T, \vec{c}_i^T)^T, \vec{y}\rangle - t_i|^p\Big) + \alpha N\\ 
&= m^+(\vec{a}) \cdot (1/2)^p + (m - m^+(\vec{a})) \cdot (3/2)^p + \alpha N\\
&\leq W \cdot (1/2)^p + (m - W) \cdot (3/2)^p + \alpha N \\
&= r^p
\; .
\end{align*}

Therefore, there exists $\vec{y}$ such that 
$\norm{B\vec{y} - \vec{t}}_p \leq r$ if and only if $(\Phi, W)$ is a `YES' instance of Max-$k$-SAT.
\end{proof}

We remark that a straightforward modification of the preceding reduction gives a reduction from an instance of Max-$k$-SAT with $k \geq 3$ on $n$ variables and $m$ clauses to a $\CVP_p$ instance of rank $n + (\lfloor \log_2(k - 2) \rfloor + 1) m$ (as opposed to rank $n + (k - 2) m$). The idea is to encode the value $k - 2$ (corresponding to the row-specific blocks $-\vec{1}_{k-2}$ used in the reduction) in binary rather than unary.

\begin{corollary}
For every $p \in [1, \infty)$ there is no $2^{o(n)}$-time algorithm for $\CVP_p$ assuming ETH.
\label{cor:eth-hardness}
\end{corollary}

\begin{proof}
The claim follows by combining the Sparsification Lemma (Proposition~\ref{prop:sparsification}) with the reduction in Theorem~\ref{thm:max-k-sat-to-high-rank}.\footnote{As a technical point, we must reduce from $k$-SAT rather than Max-$k$-SAT to show hardness under ETH because the Sparsification Lemma works for $k$-SAT.}
\end{proof}

When $k = 2$, the rank of the $\CVP_p$ instance output by the reduction in Theorem~\ref{thm:max-k-sat-to-high-rank} is $n$. 
Therefore, we get the following corollary.
\begin{corollary}
If there exists a $2^{(\omega/3 - \eps) n}$-time (resp. $2^{(1 - \eps) n}$-time and polynomial space) algorithm for $\CVP_p$ for any $p \in [1, \infty)$ and for any constant $\eps > 0$, then there exists a $2^{(\omega/3 - \eps) n}$-time (resp. $2^{(1 - \eps) n}$-time and polynomial space) algorithm for (Weighted-)Max-2-SAT and (Weighted-)Max-Cut.
\label{cor:max-cut-hardness}
\end{corollary}

\subsection{The hardness of \texorpdfstring{$\SVP_\infty$}{SVP-infinity} and \texorpdfstring{$\CVP_\infty$}{CVP-infinity}}
\label{sec:infinity!}

\begin{theorem}
There is a polynomial-time reduction from a $k$-SAT instance with $n$ variables to a $\CVP_{\infty}$ instance of rank $n$.
\label{thm:k-sat-to-cvp-infinity}
\end{theorem}

\begin{proof}
We give a reduction from a $k$-SAT formula $\Phi$ with $n$ variables $x_1, \ldots, x_n$ and $m$ clauses $C_1, \ldots, C_m$ to an instance $(B, \vec{t}, r)$ of $\CVP_{\infty}$. We assume without loss of generality that each variable appears at most once per clause. We define the output $\CVP_{\infty}$ instance as follows. Let $d := m + n$. The basis $B \in \R^{d \times n}$ in the output instance has the form

\[
B = 
\left(
 \begin{array}{c}
\vec{b}_1^T \\
\vdots \\
\vec{b}_m^T \\
(k + 1) \cdot I_n
 \end{array}
\right), \qquad
\vec{t} = 
\left(
 \begin{array}{c}
t_1 \\
\vdots \\
t_m \\
(k + 1)/2 \cdot \vec{1}_{n}
 \end{array}
\right),
\]
with rows $\vec{b}_i$ of $B$ satisfying $\vec{b}_i \in \R^n$ for $1 \leq i \leq m$. For every $1 \leq i \leq m$, set $\vec{b}_i$ as in the proof of Theorem~\ref{thm:max-k-sat-to-high-rank} and set $t_i := (k + 1)/2 - |N_i|$. Set $r := (k - 1)/2$.

We next analyze for which $\vec{y} \in \Z^n$ it holds that $\norm{B\vec{y} - \vec{t}}_{\infty} \leq r$. Given $\vec{y} \notin \set{0, 1}^n$, $\norm{B\vec{y} - \vec{t}}_{\infty} \geq \norm{(k-1) \cdot I_n \vec{y} - (k-1)/2 \cdot \vec{1}_n}_{\infty} \geq 3(k-1)/2 > r$, so we only need to analyze the case when $\vec{y} \in \set{0, 1}^n$. Consider an assignment $\vec{y} \in \set{0, 1}^n$ to the variables of $\Phi$. Then

\begin{align*}
\Big|\langle\vec{b}_i, \vec{y}\rangle - t_i\Big| &= \Big|\sum_{s \in P_i} y_{\ind(\ell_{i, s})} - \sum_{s \in N_i} y_{\ind(\ell_{i, s})} - ((k+1)/2 - |N_i|) \Big| \\
                                                        &= \Big|\sum_{s \in P_i} y_{\ind(\ell_{i, s})} - \sum_{s \in N_i} (1 - y_{\ind(\ell_{i, s})}) - (k+1)/2 \Big| \\
                                                        &= \Big||S_i(\vec{a})| - (k+1)/2 \Big|.
\end{align*}
It follows that if $|S_i(\vec{a})| = 0$ then $|\langle\vec{b}_i, \vec{y}\rangle - t_i| = (k+1)/2$, and otherwise $|\langle\vec{b}_i, \vec{y}\rangle - t_i| \leq (k - 1)/2$. Because $|S_i(\vec{a})| \geq 1$ if and only if $C_i$ is satisfied, we then have that $\norm{B\vec{y} - \vec{t}}_{\infty} = \max \set{|\langle\vec{b}_1, \vec{y}\rangle - t_1|, \ldots, |\langle\vec{b}_m, \vec{y}\rangle - t_m|, (k - 1)/2} = (k - 1)/2 = r$ if $\vec{y}$ satisfies $\Phi$, and $\norm{B\vec{y} - \vec{t}}_{\infty} = (k+1)/2 > r$ otherwise. Therefore there exists $\vec{y}$ such that $\norm{B\vec{y} - \vec{t}}_{\infty} \leq r$ if and only if $\Phi$ is satisfiable.
\end{proof}

\begin{lemma}
There is a polynomial-time reduction from a $k$-SAT instance with $n$ variables to an $\SVP_{\infty}$ instance of rank $n + 1$.
\label{lem:k-sat-to-svp-infinity}
\end{lemma}

\begin{proof}
We give a reduction from a $k$-SAT formula $\Phi$ with $n$ variables $x_1, \ldots, x_n$ and $m$ clauses $C_1, \ldots, C_m$ to an instance $(B', r)$ of $\SVP_{\infty}$. Define the basis $B' \in \R^{(m + n + 1) \times (n + 1)}$ as
\[
B' := 
\begin{pmatrix}
 B & -\vec{t} \\
 \vec{0}_{n}^T & -(k - 1)/2
\end{pmatrix},
\]
where $B$ and $\vec{t}$ are as defined in the proof of Theorem~\ref{thm:k-sat-to-cvp-infinity}, and set $r := (k - 1)/2$.
We consider for which $\vec{y} \in \Z^{n+1} \setminus \set{\vec{0}_{n+1}}$ it holds that $\norm{B\vec{y}}_{\infty} \leq r$. It is not hard to check that if $|y_i| \geq 2$ for some $1 \leq i \leq n + 1$, or if the signs of $y_i$ and $y_{n+1}$ differ for some $1 \leq i \leq n$, then $\norm{B\vec{y}}_{\infty} > r$. Therefore we need only consider $\vec{y}$ of the form $\vec{y} = \pm(\vec{a}^T, 1)^T$ where $\vec{a} \in \set{0, 1}^n$.
But for such a $\vec{y}$ we have that $\norm{B'\vec{y}}_{\infty} = \norm{B\vec{a} - \vec{t}}_{\infty}$, and $\norm{B\vec{a} - \vec{t}}_{\infty} \leq (k - 1)/2$ if and only if $\vec{a}$ is a satisfying assignment to $\Phi$ by the analysis in the proof of Theorem~\ref{thm:k-sat-to-cvp-infinity}.
\end{proof}

\begin{corollary}
	\label{cor:infinity!}
For any constant $\eps > 0$ there is no $2^{(1-\eps)n}$-time algorithm for $\SVP_{\infty}$ or $\CVP_{\infty}$ assuming SETH, and there is no $2^{o(n)}$-time algorithm for $\SVP_{\infty}$ or $\CVP_{\infty}$ assuming ETH.
\label{cor:svp-cvp-infinity-seth}
\end{corollary}

\begin{proof}
Combine Theorem~\ref{thm:k-sat-to-cvp-infinity} and Lemma~\ref{lem:k-sat-to-svp-infinity}.
\end{proof}

Note that the preceding reduction in fact achieves an approximation factor of $\gamma = \gamma(k) := 1 + 2/(k-1)$. This implies that for every constant $\eps > 0$, there is a $\gamma_\eps > 1$ such that no $2^{(1-\eps)n}$-time algorithm that approximates $\SVP_\infty$ or $\CVP_\infty$ to within a factor of $\gamma_\eps$ unless SETH fails.

Finally, we remark that the reduction given in Theorem~\ref{thm:max-k-sat-to-high-rank} is \emph{parsimonious} when used as a reduction from $2$-SAT. I.e., there is a one-to-one correspondence between satisfying assignments in the input instance and close vectors in the output instance.
The reductions given in Theorem~\ref{thm:k-sat-to-cvp-infinity} and Lemma~\ref{lem:k-sat-to-svp-infinity} are also parsimonious.\footnote{Actually, in the case of $\SVP_{\infty}$, each satisfying assignment to the SAT formula corresponds to a pair $\pm \vec{v}$ of shortest non-zero vectors, so that there are exactly twice as many such vectors as there are satisfying assignments.}

Because $\#2$-SAT is $\SP$-hard, our reductions therefore show that the counting version of $\CVP_p$ (called the Vector Counting Problem) is $\SP$-hard for all $1 \leq p \leq \infty$, and that the counting version of $\SVP_{\infty}$ is $\SP$-hard. This improves (and arguably simplifies the proof of) a result of Charles~\cite{journals/jcss/Charles07}, which showed that the counting version of $\CVP_2$ is $\SP$-hard.

\subsection*{Acknowledgments}
We would like to thank Oded Regev for many fruitful discussions and for helpful comments on an earlier draft of this work. We also thank Vinod Vaikuntanathan for recommending that we consider Gap-ETH.

\bibliographystyle{alpha}

\begin{thebibliography}{WLTB11}

\bibitem[ABSS93]{ABSS93}
Sanjeev Arora, L{\'a}szl{\'o} Babai, Jacques Stern, and Z.~Sweedyk.
\newblock The hardness of approximate optima in lattices, codes, and systems of
  linear equations.
\newblock In {\em FOCS}, 1993.

\bibitem[ABW15]{conf/focs/AbboudBW15}
Amir Abboud, Arturs Backurs, and Virginia~Vassilevska Williams.
\newblock Tight hardness results for {LCS} and other sequence similarity
  measures.
\newblock In {\em FOCS}, 2015.

\bibitem[ADPS16]{new_hope}
Erdem Alkim, L{\'e}o Ducas, Thomas P{\"o}ppelmann, and Peter Schwabe.
\newblock Post-quantum key exchange --- {A} new hope.
\newblock In {\em USENIX Security Symposium}, 2016.

\bibitem[ADRS15]{ADRS15}
Divesh Aggarwal, Daniel Dadush, Oded Regev, and Noah Stephens{-}Davidowitz.
\newblock Solving the {S}hortest {V}ector {P}roblem in $2^n$ time via discrete
  {G}aussian sampling.
\newblock In {\em STOC}, 2015.

\bibitem[ADS15]{ADS15}
Divesh Aggarwal, Daniel Dadush, and Noah Stephens{-}Davidowitz.
\newblock Solving the {C}losest {V}ector {P}roblem in $2^n$ time--- {T}he
  discrete {G}aussian strikes again!
\newblock In {\em FOCS}, 2015.

\bibitem[AJ08]{AJ08}
V.~Arvind and Pushkar~S. Joglekar.
\newblock Some sieving algorithms for lattice problems.
\newblock In {\em IARCS Annual Conference on Foundations of Software Technology
  and Theoretical Computer Science}, 2008.

\bibitem[Ajt98]{Ajtai-SVP-hard}
Mikl\'{o}s Ajtai.
\newblock The {S}hortest {V}ector {P}roblem in {L2} is {NP}-hard for randomized
  reductions.
\newblock In {\em STOC}, 1998.

\bibitem[Ajt04]{Ajtai96}
Mikl{\'o}s Ajtai.
\newblock Generating hard instances of lattice problems.
\newblock In {\em Complexity of computations and proofs}, volume~13 of {\em
  Quad. Mat.}, pages 1--32. Dept. Math., Seconda Univ. Napoli, Caserta, 2004.
\newblock Preliminary version in STOC'96.

\bibitem[AKKV11]{AlekhnovichKKV11}
Mikhail Alekhnovich, Subhash Khot, Guy Kindler, and Nisheeth~K. Vishnoi.
\newblock Hardness of approximating the {C}losest {V}ector {P}roblem with
  {P}re-processing.
\newblock {\em Computational Complexity}, 20, 2011.

\bibitem[AKS01]{AKS01}
Mikl\'{o}s Ajtai, Ravi Kumar, and D.~Sivakumar.
\newblock A sieve algorithm for the {S}hortest {L}attice {V}ector {P}roblem.
\newblock In {\em STOC}, 2001.

\bibitem[AKS02]{AKS02}
Mikl{\'o}s Ajtai, Ravi Kumar, and D.~Sivakumar.
\newblock Sampling short lattice vectors and the {C}losest {L}attice {V}ector
  {P}roblem.
\newblock In {\em CCC}, 2002.

\bibitem[APS15]{APS15}
Martin~R. Albrecht, Rachel Player, and Sam Scott.
\newblock On the concrete hardness of {L}earning with {E}rrors.
\newblock {\em J. Mathematical Cryptology}, 9(3):169--203, 2015.

\bibitem[AR05]{AharonovR05}
Dorit Aharonov and Oded Regev.
\newblock Lattice problems in {NP} $\cap$ {coNP}.
\newblock {\em J.~ACM}, 52(5):749--765, 2005.
\newblock Preliminary version in FOCS 2004.

\bibitem[BCD{\etalchar{+}}16]{frodo}
Joppe~W. Bos, Craig Costello, L{\'{e}}o Ducas, Ilya Mironov, Michael Naehrig,
  Valeria Nikolaenko, Ananth Raghunathan, and Douglas Stebila.
\newblock Frodo: Take off the ring! {P}ractical, quantum-secure key exchange
  from {LWE}.
\newblock In {\em CCS}, 2016.

\bibitem[BD15]{BonifasD14}
Nicolas Bonifas and Daniel Dadush.
\newblock Short paths on the {V}oronoi graph and the {C}losest {V}ector
  {P}roblem with {P}reprocessing.
\newblock In {\em SODA}, 2015.

\bibitem[BDGL16]{BDGL16}
Anja Becker, L{\'e}o Ducas, Nicolas Gama, and Thijs Laarhoven.
\newblock New directions in nearest neighbor searching with applications to
  lattice sieving.
\newblock In {\em SODA}, 2016.

\bibitem[BI15]{conf/stoc/BackursI15}
Arturs Backurs and Piotr Indyk.
\newblock {E}dit {D}istance cannot be computed in strongly subquadratic time
  (unless {SETH} is false).
\newblock In {\em STOC}, 2015.

\bibitem[BN09]{BN09}
Johannes Bl{\"o}mer and Stefanie Naewe.
\newblock Sampling methods for shortest vectors, closest vectors and successive
  minima.
\newblock {\em Theoret. Comput. Sci.}, 410(18):1648--1665, 2009.

\bibitem[CDL{\etalchar{+}}12]{conf/coco/CyganDLMNOPSW12}
Marek Cygan, Holger Dell, Daniel Lokshtanov, D{\'{a}}niel Marx, Jesper
  Nederlof, Yoshio Okamoto, Ramamohan Paturi, Saket Saurabh, and Magnus
  Wahlstr{\"{o}}m.
\newblock On problems as hard as {CNF-SAT}.
\newblock In {\em CCC}, 2012.

\bibitem[CFK{\etalchar{+}}15]{books/sp/CyganFKLMPPS15}
Marek Cygan, Fedor~V. Fomin, Lukasz Kowalik, Daniel Lokshtanov, D{\'{a}}niel
  Marx, Marcin Pilipczuk, Michal Pilipczuk, and Saket Saurabh.
\newblock {\em Parameterized Algorithms}.
\newblock Springer, 2015.

\bibitem[Cha07]{journals/jcss/Charles07}
Denis~Xavier Charles.
\newblock Counting lattice vectors.
\newblock {\em J. Comput. Syst. Sci.}, 73(6):962--972, 2007.

\bibitem[CLR{\etalchar{+}}14]{conf/soda/ChechikLRSTW14}
Shiri Chechik, Daniel~H. Larkin, Liam Roditty, Grant Schoenebeck, Robert~Endre
  Tarjan, and Virginia~Vassilevska Williams.
\newblock Better approximation algorithms for the graph diameter.
\newblock In {\em SODA}, 2014.

\bibitem[CN98]{CN98}
Jin-Yi Cai and Ajay Nerurkar.
\newblock Approximating the {SVP} to within a factor $(1+1/\dim^\eps)$ is
  {NP}-hard under randomized conditions.
\newblock In {\em CCC}, 1998.

\bibitem[Dad12]{Dadush2012}
Daniel Dadush.
\newblock A {$O(1/\epsilon^2)^n$}-time sieving algorithm for approximate
  {I}nteger {P}rogramming.
\newblock In {\em LATIN}, 2012.

\bibitem[Din16]{journals/eccc/Dinur16}
Irit Dinur.
\newblock Mildly exponential reduction from gap {3SAT} to polynomial-gap
  label-cover.
\newblock {\em Electronic Colloquium on Computational Complexity {(ECCC)}},
  23:128, 2016.

\bibitem[DKRS03]{DKRS03}
Irit Dinur, Guy Kindler, Ran Raz, and Shmuel Safra.
\newblock Approximating {CVP} to within almost-polynomial factors is {NP}-hard.
\newblock {\em Combinatorica}, 23(2):205--243, 2003.

\bibitem[DPV11]{DPV11}
Daniel Dadush, Chris Peikert, and Santosh Vempala.
\newblock Enumerative lattice algorithms in any norm via {M}-ellipsoid
  coverings.
\newblock In {\em FOCS}, 2011.

\bibitem[DRS14]{cvpp}
Daniel Dadush, Oded Regev, and Noah Stephens{-}Davidowitz.
\newblock On the {C}losest {V}ector {P}roblem with a distance guarantee.
\newblock In {\em {CCC}}, pages 98--109, 2014.

\bibitem[FM04]{FeigeMicciancio}
Uriel Feige and Daniele Micciancio.
\newblock The inapproximability of lattice and coding problems with
  preprocessing.
\newblock {\em Journal of Computer and System Sciences}, 69(1):45--67, 2004.
\newblock Preliminary version in CCC 2002.

\bibitem[GG00]{GG00}
Oded Goldreich and Shafi Goldwasser.
\newblock On the limits of nonapproximability of lattice problems.
\newblock {\em J.~Comput. Syst. Sci.}, 60(3):540--563, 2000.
\newblock Preliminary version in STOC 1998.

\bibitem[GHNR03]{journals/dam/GrammHNR03}
Jens Gramm, Edward~A. Hirsch, Rolf Niedermeier, and Peter Rossmanith.
\newblock Worst-case upper bounds for {MAX-2-SAT} with an application to
  {MAX-CUT}.
\newblock {\em Discrete Applied Mathematics}, 130(2):139--155, 2003.

\bibitem[GJS76]{GJS76}
Michael~R. Garey, David~S. Johnson, and Larry Stockmeyer.
\newblock Some simplified {NP}-complete graph problems.
\newblock {\em Theoretical Computer Science}, 1(3):237--267, 1976.

\bibitem[GMSS99]{GMSS99}
Oded Goldreich, Daniele Micciancio, Shmuel Safra, and Jean-Pierre Seifert.
\newblock Approximating shortest lattice vectors is not harder than
  approximating closest lattice vectors.
\newblock {\em Information Processing Letters}, 71(2):55 -- 61, 1999.

\bibitem[GN08]{GN08}
Nicolas Gama and Phong~Q. Nguyen.
\newblock Finding short lattice vectors within {M}ordell's inequality.
\newblock In {\em {STOC}}, 2008.

\bibitem[GPV08]{GPV08}
Craig Gentry, Chris Peikert, and Vinod Vaikuntanathan.
\newblock Trapdoors for hard lattices and new cryptographic constructions.
\newblock In {\em STOC}, pages 197--206, 2008.

\bibitem[HP]{hp/chernoff}
S.~Har-Peled.
\newblock {Concentration of Random Variables -- Chernoff's Inequality}.
\newblock Available at
  \url{http://sarielhp.org/teach/13/b_574_rand_alg/lec/07_chernoff.pdf}.

\bibitem[HR12]{HRsvp}
Ishay Haviv and Oded Regev.
\newblock Tensor-based hardness of the {S}hortest {V}ector {P}roblem to within
  almost polynomial factors.
\newblock {\em Theory of Computing}, 8(23):513--531, 2012.
\newblock Preliminary version in STOC'07.

\bibitem[IP99]{IP1999}
Russell Impagliazzo and Ramamohan Paturi.
\newblock On the complexity of $k$-{SAT}.
\newblock In {\em CCC}, pages 237--240, 1999.

\bibitem[IPZ01]{journals/jcss/ImpagliazzoPZ01}
Russell Impagliazzo, Ramamohan Paturi, and Francis Zane.
\newblock Which problems have strongly exponential complexity?
\newblock {\em J. Comput. Syst. Sci.}, 63(4):512--530, 2001.

\bibitem[Jam06]{jameson_2006}
G.~J.~O. Jameson.
\newblock Counting zeros of generalised polynomials: {D}escartes' rule of signs
  and {L}aguerre's extensions.
\newblock {\em The Mathematical Gazette}, 90(518):223--234, 2006.

\bibitem[JS98]{JS98}
Antoine Joux and Jacques Stern.
\newblock Lattice reduction: A toolbox for the cryptanalyst.
\newblock {\em Journal of Cryptology}, 11(3):161--185, 1998.

\bibitem[Kan87]{Kannan87}
Ravi Kannan.
\newblock Minkowski's convex body theorem and {I}nteger {P}rogramming.
\newblock {\em Math. Oper. Res.}, 12(3):415--440, 1987.

\bibitem[Kho05]{Khot05svp}
Subhash Khot.
\newblock Hardness of approximating the {S}hortest {V}ector {P}roblem in
  lattices.
\newblock {\em Journal of the ACM}, 52(5):789--808, September 2005.
\newblock Preliminary version in FOCS'04.

\bibitem[KPV14]{KhotPV12}
Subhash Khot, Preyas Popat, and Nisheeth~K. Vishnoi.
\newblock $2^{\log^{1-\eps} n}$ hardness for {C}losest {V}ector {P}roblem with
  {P}reprocessing.
\newblock {\em SIAM Journal on Computing}, 43(3):1184--1205, 2014.

\bibitem[Laa15]{Laarhoven2015}
Thijs Laarhoven.
\newblock Sieving for shortest vectors in lattices using angular
  locality-sensitive hashing.
\newblock In {\em CRYPTO}, 2015.

\bibitem[Len83]{Lenstra83}
Hendrik~W. Lenstra, Jr.
\newblock {I}nteger {P}rogramming with a fixed number of variables.
\newblock {\em Math. Oper. Res.}, 8(4):538--548, 1983.

\bibitem[LG14]{LeGall14}
Fran\c{c}ois Le~Gall.
\newblock Powers of tensors and fast matrix multiplication.
\newblock In {\em ISAAC}, 2014.

\bibitem[LLL82]{LLL82}
Arjen~K. Lenstra, Hendrik~W. Lenstra, Jr., and L{\'{a}}szl{\'{o}} Lov{\'a}sz.
\newblock Factoring polynomials with rational coefficients.
\newblock {\em Math. Ann.}, 261(4):515--534, 1982.

\bibitem[LWXZ11]{LWXZ11}
Mingjie Liu, Xiaoyun Wang, Guangwu Xu, and Xuexin Zheng.
\newblock Shortest lattice vectors in the presence of gaps.
\newblock {\em IACR Cryptology ePrint Archive}, 2011:139, 2011.

\bibitem[Mic01a]{MicciancioCVPP}
Daniele Micciancio.
\newblock The hardness of the {C}losest {V}ector {P}roblem with
  {P}reprocessing.
\newblock {\em IEEE Transactions on Information Theory}, 47(3):1212--1215,
  2001.

\bibitem[Mic01b]{Mic01svp}
Daniele Micciancio.
\newblock The {S}hortest {V}ector {P}roblem is {NP}-hard to approximate to
  within some constant.
\newblock {\em SIAM Journal on Computing}, 30(6):2008--2035, March 2001.
\newblock Preliminary version in FOCS 1998.

\bibitem[MR17]{conf/icalp/ManurangsiR17}
Pasin Manurangsi and Prasad Raghavendra.
\newblock A birthday repetition theorem and complexity of approximating dense
  csps.
\newblock In {\em 44th International Colloquium on Automata, Languages, and
  Programming, {ICALP} 2017, July 10-14, 2017, Warsaw, Poland}, pages
  78:1--78:15, 2017.

\bibitem[MV10]{MV10}
Daniele Micciancio and Panagiotis Voulgaris.
\newblock Faster exponential time algorithms for the {S}hortest {V}ector
  {P}roblem.
\newblock In {\em SODA}, pages 1468--1480, 2010.

\bibitem[MV13]{MV13}
Daniele Micciancio and Panagiotis Voulgaris.
\newblock A deterministic single exponential time algorithm for most lattice
  problems based on {V}oronoi cell computations.
\newblock {\em SIAM Journal on Computing}, 42(3):1364--1391, 2013.

\bibitem[MW15]{MicciancioWalter14}
Daniele Micciancio and Michael Walter.
\newblock Fast lattice point enumeration with minimal overhead.
\newblock In {\em SODA}, 2015.

\bibitem[MW16]{MW16}
Daniele Micciancio and Michael Walter.
\newblock Practical, predictable lattice basis reduction.
\newblock In {\em Eurocrypt}, 2016.

\bibitem[NIS16]{NIST_quantum}
{NIST} post-quantum standardization call for proposals.
\newblock
  \url{http://csrc.nist.gov/groups/ST/post-quantum-crypto/cfp-announce-dec2016.html},
  2016.
\newblock Accessed: 2017-04-02.

\bibitem[NS01]{NS01}
Phong~Q. Nguyen and Jacques Stern.
\newblock The two faces of lattices in cryptology.
\newblock In {\em Cryptography and Lattices}, pages 146--180. Springer, 2001.

\bibitem[NV08]{NguyenVidick08}
Phong~Q. Nguyen and Thomas Vidick.
\newblock Sieve algorithms for the shortest vector problem are practical.
\newblock {\em J. Math. Cryptol.}, 2(2):181--207, 2008.

\bibitem[Odl90]{Odl90}
Andrew~M. Odlyzko.
\newblock The rise and fall of knapsack cryptosystems.
\newblock {\em Cryptology and Computational Number Theory}, 42:75--88, 1990.

\bibitem[Pei08]{Peikert08}
Chris Peikert.
\newblock Limits on the hardness of lattice problems in $\ell_p$ norms.
\newblock {\em Computational Complexity}, 17(2):300--351, May 2008.
\newblock Preliminary version in CCC 2007.

\bibitem[Pei16]{chris_survey}
Chris Peikert.
\newblock A decade of lattice cryptography.
\newblock {\em Foundations and Trends in Theoretical Computer Science},
  10(4):283--424, 2016.

\bibitem[PS09]{PS09}
Xavier Pujol and Damien Stehl{\'e}.
\newblock Solving the {S}hortest {L}attice {V}ector {P}roblem in time $2^{2.465
  n}$.
\newblock {\em IACR Cryptology ePrint Archive}, 2009:605, 2009.

\bibitem[PW10]{conf/soda/PatrascuW10}
Mihai P{\u{a}}tra{\c{s}}cu and Ryan Williams.
\newblock On the possibility of faster {SAT} algorithms.
\newblock In {\em SODA}, 2010.

\bibitem[Reg04]{Regev03B}
Oded Regev.
\newblock Improved inapproximability of lattice and coding problems with
  preprocessing.
\newblock {\em IEEE Transactions on Information Theory}, 50(9):2031--2037,
  2004.
\newblock Preliminary version in CCC'03.

\bibitem[Reg09]{Reg09}
Oded Regev.
\newblock On lattices, learning with errors, random linear codes, and
  cryptography.
\newblock {\em Journal of the ACM}, 56(6):Art. 34, 40, 2009.

\bibitem[RR06]{RR06}
Oded Regev and Ricky Rosen.
\newblock Lattice problems and norm embeddings.
\newblock In {\em {STOC}}, 2006.

\bibitem[Sch87]{Schnorr87}
Claus~P. Schnorr.
\newblock A hierarchy of polynomial time lattice basis reduction algorithms.
\newblock {\em Theoretical Computer Science}, 53:201 -- 224, 1987.

\bibitem[SFS09]{SFS09}
Naftali Sommer, Meir Feder, and Ofir Shalvi.
\newblock Finding the closest lattice point by iterative slicing.
\newblock {\em SIAM J. Discrete Math.}, 23(2):715--731, 2009.

\bibitem[Sha84]{Shamir84}
Adi Shamir.
\newblock A polynomial-time algorithm for breaking the basic {M}erkle-{H}ellman
  cryptosystem.
\newblock {\em IEEE Trans. Inform. Theory}, 30(5):699--704, 1984.

\bibitem[Vai15]{priv:Vinod}
Vinod Vaikuntanathan.
\newblock {P}rivate communication, 2015.

\bibitem[vEB81]{Boas81}
Peter van Emde~Boas.
\newblock Another {NP}-complete problem and the complexity of computing short
  vectors in a lattice.
\newblock Technical report, University of Amsterdam, Department of Mathematics,
  Netherlands, 1981.
\newblock Technical Report 8104.

\bibitem[Wil05]{Williams05}
Ryan Williams.
\newblock A new algorithm for optimal 2-{C}onstraint {S}atisfaction and its
  implications.
\newblock {\em Theoretical Computer Science}, 348(2-3):357--365, 2005.

\bibitem[Wil12]{Williams12-fast_matrix}
Virginia~Vassilevska Williams.
\newblock Multiplying matrices faster than {C}oppersmith-{W}inograd.
\newblock In {\em STOC}, 2012.

\bibitem[Wil15]{conf/iwpec/Williams15}
Virginia~Vassilevska Williams.
\newblock Hardness of easy problems: Basing hardness on popular conjectures
  such as the {S}trong {E}xponential {T}ime {H}ypothesis (invited talk).
\newblock In {\em {IPEC}}, pages 17--29, 2015.

\bibitem[Wil16]{conf/coco/Williams16}
Ryan Williams.
\newblock Strong {ETH} breaks with {M}erlin and {A}rthur: Short non-interactive
  proofs of batch evaluation.
\newblock In {\em CCC}, 2016.

\bibitem[WLTB11]{WLTB11}
Xiaoyun Wang, Mingjie Liu, Chengliang Tian, and Jingguo Bi.
\newblock Improved {N}guyen-{V}idick heuristic sieve algorithm for shortest
  vector problem.
\newblock In {\em ASIACCS}, 2011.

\bibitem[Woe08]{Woeginger08}
Gerhard~J. Woeginger.
\newblock Open problems around exact algorithms.
\newblock {\em Discrete Applied Mathematics}, 156(3):397--405, 2008.

\end{thebibliography}
\newcommand{\etalchar}[1]{$^{#1}$}

\end{document}